\documentclass[11pt, authoryear]{elsarticle}

\usepackage{amssymb}
\usepackage{amsthm}
\usepackage{amsmath}
\usepackage{amsfonts}
\usepackage{mathtools}
\usepackage{enumitem}
\usepackage[hang,flushmargin]{footmisc}
\usepackage[left=2cm, right=2cm, top=2cm]{geometry} 
\usepackage{theoremref}
\usepackage{booktabs}
\usepackage[nolists]{endfloat}
\usepackage{xcolor}
\usepackage{pdflscape}
\usepackage{multirow}
\usepackage{threeparttable}
\usepackage{caption}

\usepackage{titlesec}
\titleformat*{\section}{\Large\bfseries}
\titleformat*{\subsection}{\large\bfseries}
\titleformat*{\subsubsection}{\bfseries}
\titlespacing*{\section}{0pt}{0.3\baselineskip}{0.3\baselineskip}
\titlespacing*{\subsection}{0pt}{0.3\baselineskip}{0.3\baselineskip}
\titlespacing*{\subsubsection}{0pt}{0.3\baselineskip}{0.3\baselineskip}

\DeclareMathOperator{\Tr}{Tr}
\DeclareMathOperator{\argmin}{arg \ min}
\DeclarePairedDelimiter\floor{\lfloor}{\rfloor}

\def\sumTau{\sum_{\substack{s=\tau T-\floor*{Th} \\ s \neq \tau T}}^{\tau T+\floor*{Th}}}

\newtheorem{theorem}{Theorem}
\newtheorem{proposition}{Proposition}
\newtheorem{lemma}{Lemma}

\allowdisplaybreaks 

\begin{document}

\thispagestyle{empty}

\begin{center}

\renewcommand{\thefootnote}{\fnsymbol{footnote}}

{\huge Time-varying Forecast Combination for High-Dimensional Data\footnote[1]{We thank Xu Cheng, Francis X. Diebold, Atsushi Inoue, Frank Schorfheide, Nese Yildiz and seminar participants at the
University of Pennsylvania, University of Rochester, the 2019 Southern Economic Society meeting and the 12th Econometric Society World Congress for their useful comments and discussions. Any remaining errors are
solely ours. }}\bigskip

\bigskip

$%
\begin{array}{c}

\text{{\large Bin Chen}} \\ 
\text{University of Rochester} \\ 
\text{ {\large Kenwin Maung}} \\ 
\text{University of Rochester}%
\end{array}%
$\bigskip

\bigskip

\bigskip
\end{center}

\textit{Abstract: } In this paper, we propose a new nonparametric estimator of time-varying forecast combination weights. When the number of individual forecasts is small, we study the asymptotic properties of the local linear estimator. When the number of candidate forecasts exceeds or diverges with the sample size, we consider penalized local linear estimation with the group SCAD penalty. We show that the estimator exhibits the oracle property and correctly selects relevant forecasts with probability approaching one. Simulations indicate that the proposed estimators outperform existing combination schemes when structural changes exist. Two empirical studies on inflation forecasting and equity premium prediction highlight the merits of our approach relative to other popular methods. \bigskip \bigskip


\bigskip

\bigskip

\verb||

\verb||

\pagebreak \setcounter{page}{1}

\section{Introduction}

Multiple forecasts of the same variable are often available to decision makers. As pointed out in the seminal paper by \cite{bates1969combination}, combinations of individual forecasts can outperform individual forecasts as economic systems are highly complex and even the most sophisticated model is likely to be misspecified. It is unlikely that a single model will dominate uniformly. Even if such a best model exists, it is very difficult to identify it in practice since many forecasts might have similar predictive accuracy. 

Structural breaks in predictive relationships pose additional challenges when generating out-of-sample forecasts. Individual forecasts may vary with structural changes caused by changes in preferences, institutional evolution or technological progress, among other reasons. It is likely that the combination of forecasts from models with different degrees of adaptability would average out individual effects and outperform forecasts from one specific model. Forecast combination can thus be viewed as a strategy against potential structural changes, in the spirit of portfolio hedging, by offering diversification gains. \citep[see e.g.][]{aiolfi2006persistence, timmermann2006forecast, elliott2005optimal}. Understandably, this ability to deal with both model uncertainty and structural changes in forecasting has motivated many authors to apply forecast combination in various fields, ranging from macroeconomics \citep{elliott2005optimal, stock2004combination} to empirical asset pricing \citep{lin2018forecasting, rapach2010out}, with many reporting significant performance gains over prevailing methods. 

Given that the relative performance of different forecasts is likely to change over time, it is natural to consider forecast combination with time-varying weights. Time-varying forecast combination was first proposed  by \cite{bates1969combination}, who developed several adaptive estimation schemes for time-varying weights based on exponential discounting or rolling estimation. In the regression context, \cite{diebold1987structural} generalize these schemes to select combination weights that minimize the weighted average of forecast errors. \cite{deutsch1994combination} and \cite{elliott2005optimal} consider a parametric interpretation of the time-varying combination weights by allowing them to be driven by smooth transitions or switching. In an empirical study, \cite{lin2018forecasting} consider the iterated mean combination and the iterated weighted combination, which improve on existing forecast combination schemes by combining them with the historical sample mean forecast. These methods rely on either rolling estimation with a fixed window size or impose certain parametric functional form assumption on the combination weights, which may be restrictive. Therefore, it is desirable to develop an alternative time-varying combination scheme which can hedge against structural changes of an unknown form.

Recently, nonparametric time-varying parameter models have proven to be a reliable tool in identifying the smoothly-varying coefficient functions. Furthermore, it has been shown to adequately capture the evolutionary behavior of economic relationships through various empirical applications. Such models were first introduced by \cite{robinson1989nonparametric, robinson1991time} and further studied by \cite{cai2007trending}, \cite{chen2012testing}, \cite{kristensen2012non}, \cite{zhang2012inference}, \cite{dahlhaus2019towards}, \cite{hongsunwang} among many others. One advantage of the nonparametric
time-varying parameter model is that little restriction is imposed on the functional forms of coefficients, apart from the regularity condition that they evolve smoothly over time. Motivated by this flexibility, we will adopt this framework in estimating the time-varying combination weights.

This paper develops two time-varying forecast combination schemes. When the number of forecasts is small, we consider a new nonparametric estimator for the combination weights and study its asymptotic properties. Our framework is general enough to accommodate some degree of non-stationarity, in particular that of local stationarity, and thus we can avoid taking a hard stance on the time series behavior of forecasts. As such, our estimator can be viewed as a generalization of the classical \cite{granger1984improved} regression estimator. To implement our nonparametric combination scheme, we consider a cross-validation (CV) bandwidth selection method and show that the selected bandwidth converges to the theoretical optimal bandwidth, which minimizes the integrated mean squared combined forecast errors (IMSCFE). Bandwidth selection here is analogous to the optimal selection of window size for classical rolling regression \citep[see for e.g.][]{hongsunwang}.

When the number of potential forecasts is allowed to be at the same order as or even larger than the sample size, we consider a two-stage penalized local linear procedure similar to that of \cite{li2015model} in studying varying coefficient models. Model selection has been an increasingly important topic in econometrics and statistics in the past twenty years, and various penalized likelihood or least-square methods have been studied to handle model selection for high-dimensional data. Examples of
commonly-used penalization schemes include the Lasso 
\citep{tibshirani1996regression}, smoothly clipped absolute deviation (SCAD) \citep{fan2001variable}, group Lasso \citep{yuan2006model}, adaptive Lasso \citep{zou2006adaptive}, and the minimax concave penalty (MCP) \citep{zhang2010nearly}. In a related context, model selection for functional coefficient models under the $i.i.d.$ assumption are considered in \cite{wang2009shrinkage}, \cite{wei2011variable} and \cite{li2015model}, and model selection for
high-dimsional linear time series models are studied in \cite{kock2015oracle}, \cite{han2020high} and \cite{diebold2019machine}. We contribute to this growing literature of high-dimensional model selection in time series econometrics by investigating the asymptotic properties of our two-stage estimator and showing that it possesses the oracle property. 

Our proposed approach has a number of appealing features. First, the forecast combination weights are modeled as some nonparametric function of time. Unlike \cite{deutsch1994combination} and \cite{elliott2005optimal}, we do not impose parametric assumptions on the functional form of time variation. Second, the finite-dimensional case and the high-dimensional cases are studied in a unified framework. The two-step scheme shrinks the weights of irrelevant forecasts to 0 and provides a practical tool to reduce dimensionality. Third, unlike the majority of the literature on forecast combinations, we investigate the asymptotic properties of our estimator, and establish results on both estimation and selection consistency.

The rest of the paper is organized as follows. In Section 2, we introduce the framework of our nonparametric time-varying combination scheme and develop the estimator when the number of forecasts is small. Section 3 derives the asymptotic properties of the estimator and the CV-selected bandwidth. Section 4 introduces the two-stage penalized estimator for combination weights in high-dimensional forecast combination. Section 5 discusses its implementation and Section 6 studies the oracle property of the estimator. In Section 7, a simulation study is conducted to assess the reliability of the low- and high-dimensional estimators in finite samples. Two empirical examples on inflation forecasting and equity premium prediction are used to illustrate the merits of our approaches in Section 8. Main mathematical proofs are collected in the appendix, while the proofs of some technical results and additional simulations are contained in an online appendix \citep{onlineappendix}.

\section{Nonparametric Forecast Combination}

\label{nonpar}

Assume that a decision maker is interested in predicting some univariate
series $y_{t+1}$, conditional on $I_{t},$ the information available at time $%
t$, which consists a set of individual forecasts $f_{t}=%
\left(f_{t1,}f_{t2},...,f_{td}\right) ^{\top}$ in addition to current and
past values of $y$, i.e. $I_{t}=\left( y_{s},f_{s}\right)_{s=1}^{t}.$ The
vast majority of studies in the forecasting literature considers a linear
forecast combination model: 
\begin{equation*}
y_{t+1}=\omega _{0}+f_{t}^{\top }\omega _{1}+\varepsilon _{t+1},\qquad
t=1,...,T,
\end{equation*}
where $\omega _{0}$ is an intercept and $\omega _{1}$ is a $d\times 1$ vector,
whose component $\omega _{1i}$ can be viewed as the weight assigned to the $%
i^{th}$ forecast, where $i=1,...,d.$ Note that the sum of $\omega _{1i}$
need not be unity.\footnote{%
Alternatively, a transformed model, which imposes the constraint that $%
\sum_{i=1}^{d}\omega _{1i}=1$ and $\omega _{0}=0$ can be considered. Namely, 
$y_{t+1}-f_{td}=\sum_{i=1}^{d-1}\omega _{1i}\left( f_{ti}-f_{td}\right)
+\varepsilon _{t+1}.$} Individual forecasting models can be viewed as
local approximations to the true data generating process (DGP) and their forecast
ability is likely to change over time due to the prevalence of structural
changes. Therefore, we consider forecast combination with time-varying
weights: 
\begin{equation*}
y_{t+1}=\omega _{0t}+f_{t}^{\top }\omega _{1t}+\varepsilon _{t+1},\qquad
t=1,...,T,
\end{equation*}
where $\left( \omega _{0t},\omega _{1t}^{\top }\right) ^{\top }$ are adapted
to the current information set $I_{t}$.

We opt to estimate the model without imposing any parametric functional form on the time variation of combination weights\footnote{There are at least three ways to estimate the time-varying weights \citep{elliott2005optimal, timmermann2006forecast}. The first method is based on a rolling window estimation with some fixed window length $c$, where $c$ is often selected arbitrarily in empirical studies. The second method assumes the form of a time-varying parameter model, where the combination weights
are assumed to follow a multivariate unit root process. The third method assumes that
weights are driven by switching \citep{elliott2005optimal} or smooth transitions \citep{deutsch1994combination} with some observed or latent state
variable. }. Specifically, we adopt the following framework of the nonparametric time-varying parameter model:
\begin{equation}
y_{t+1}=\omega _{0}\left( t/T\right) +f_{t}^{\top }\omega _{1}\left(
t/T\right) +\varepsilon_{t+1},\qquad t=1,...,T,  \label{yt1}
\end{equation}
where $\omega_{0}:[0,1]\rightarrow \mathbb{R}^{1}$ and $\omega_{1}:[0,1] \rightarrow \mathbb{R}^{d}$ are smooth functions of the standardized time $t/T$ over [0,1]. This model was introduced by \cite{robinson1989nonparametric, robinson1991time} and has been studied extensively. The specification that $\omega _{jt} \equiv \omega_j(t/T)$, for $j=0,1$, are functions of the ratio $t/T$ rather than time $t$ itself is a common scaling scheme in the literature which guarantees that the amount of local information increases suitably with the sample size. One difficulty regarding estimation of a predictive model is that no symmetric data is available when making a forecast at any time point $t$. In other words, we are unable to use data from $t+1$ onwards when producing forecasts at time $t$. In the context of nonparametric regression, this is essentially the boundary problem. Note that although local linear smoothing can enhance the convergence rate of the asymptotic bias in the boundary, the asymptotic variance at a boundary point is inevitably larger  because we have fewer observations contributing to the estimator on a smaller data interval. To further reduce the variance, we adopt the reflection method following Hall and Wehrly (1991) and Chen and Hong (2012). Specifically, we reflect the data at each data point $t$ and obtain pseudodata $(y_{s+1},f_{s}^{\top})=(y_{2t-s+1},f_{2t-s}^{\top })$ for $t+1\leq s\leq t+\lfloor Th\rfloor,$ where $\lfloor Th\rfloor $ denotes the integer part of $Th$ and $h$ is the bandwidth used in estimation. We use the synthesized data (the union of the original data and pseudodata) to estimate $\beta _{t}=( \omega _{0t},\omega _{1t}^{^{\top }}) ^{\top }$ for each $t$ via local linear estimation.

Let $z_{st}=\left( 1,\frac{s-t}{T}\right) ^{\top }$ and $k_{st}=h^{-1}k \left( \frac{s-t}{Th}\right) ,$ where the kernel $k :[ -1,1] \rightarrow \mathbb{R}^{+}$ is a prespecified symmetric probability density. Examples of $k(\cdot )$ include the uniform, Epanechnikov and quartic kernels. As discussed in \cite{hongsunwang}, $\lfloor Th\rfloor $ is analogous to the window length of a rolling window regression. The local linear parameter estimator for combination weights at time $t$ is obtained by minimizing the local sum of squared residuals: 
\begin{equation}
\underset{\gamma \in \mathbb{R}^{2(d+1)}}{\min }T^{-1}\sum_{s=t-\lfloor
Th\rfloor ,s\neq t}^{t+\lfloor Th\rfloor }k_{st}\left[ y_{s+1}-\alpha
_{0}^{\top }x_{s}-\alpha _{1}^{\top }\left( \frac{s-t}{T}\right) x_{s}\right]
^{2}=T^{-1}\sum_{s=t-\lfloor Th\rfloor , s\neq t }^{t+\lfloor Th\rfloor
}k_{st}(y_{s+1}-\gamma ^{\top }q_{st})^{2},  \label{LLeast}
\end{equation}%
where $\gamma =(\alpha _{0}^{\top },\alpha _{1}^{\top })^{\top }$ is a $%
2(d+1)\times 1$ vector, $\alpha _{j}$ is a $(d+1)\times 1$ coefficient
vector for $(\frac{s-t}{T})^{j}x_{s},$ $j=0,1,$ $q_{st}=z_{st}\otimes x_{s}$
is a $2(d+1)\times 1$ vector, and $\otimes $ is the Kronecker product.

Minimizing \eqref{LLeast} with respect to $\gamma_t$ yields the local linear
estimate of $\beta(t/T)$,
\begin{equation}
\hat{\beta}_t = \hat{\beta}(t/T) = (e_1^\top \otimes I_{(d+1)}) \hat{\gamma}_t,  \label{betasol}
\end{equation}
where $e_1 = (1,0)^\top$, $I_{(d+1)}$ is a $(d+1) \times (d+1)$ identity
matrix, and
\begin{equation}
\hat{\gamma}_t = \bigg(\sum_{ s=t-\floor*{Th}  , s \neq t}^{t+%
\floor*{Th}} k_{st} q_{st} q_{st}^\top \bigg)^{-1} \sum_{ s=t-%
\floor*{Th} , s \neq t}^{t+\floor*{Th}} k_{st} q_{st} y_{s+1}.
\label{alter}
\end{equation}

The estimator takes the form of a leave-one-out local
linear estimator considered in \citet{chen2012testing} because of the
predictive design of the regression, which refers to the fact that we do not observe $y_{t+1}$ at time $t$ and hence cannot use it in the estimation.

\section{Asymptotic Properties}

\label{sect.asym}

In this section, we consider the asymptotic properties of the estimator when the number of forecasts is relatively small ($d \ll T$) and thus no regularization is required. To begin, we impose the following regularity conditions.

\bigskip

\noindent \textbf{Assumption A.1 (Mixing condition):} The process $\{R_t\}= (\varepsilon_{t+1}, X_t^\top)^\top$ is a $\beta$-mixing process with mixing coefficient ${\beta^*(j)}$ satisfying  $\sum_{l=1}^\infty l^4 \beta^*(l)^{\delta/(1+\delta)} <
\infty$ for some $\delta > 0 $.

\bigskip

\noindent \textbf{Assumption A.2 (Moment conditions):} The following moment
conditions are satisfied: \setlist{nolistsep}

\begin{enumerate}[label=(\roman*), noitemsep]

\item $\sup_{1 \leq t \leq T} E \| R_t \|^{4+\delta} < \infty$ for some $%
\delta >0$,

\item $\beta_t = \beta(t/T)$ is a smooth function such that its second order
derivative is continuous in $[0,1]$,

\item $M(t/T) = E(X_t X_t^\top)$, $\sigma^2(t/T) = E(\varepsilon_t^2)$ and $%
V(t/T) = E(X_t X_t^\top \varepsilon_{t+1}^2)$, where $M(\tau)$, $%
\sigma^2(\tau)$ and $V(\tau)$ are Lipschitz continuous for all $\tau \in
[0,1]$, and $M(\tau)$ is positive definite.
\end{enumerate}
\bigskip

\noindent \textbf{Assumption A.3 (Forecast error):} Let $\{ \varepsilon_t\}$
be a martingale difference sequence (m.d.s.). In particular, $%
E(\varepsilon_{t+1}|\mathcal{I}_t) = 0$, where $\mathcal{I}_t = \{X_t^\top,
X_{t-1}^\top, \ldots, \varepsilon_{t}, \varepsilon_{t-1}, \ldots \}$.

\bigskip

\noindent \textbf{Assumption A.4 (Kernel):} $k: [-1,1] \rightarrow \mathbb{R}%
^+$ is a symmetric bounded probability density function. Further, assume $%
k(0) \geq k(u)$ for all $u \in [-1,1]$, and $\int k^2(u) du < \infty$.

\bigskip

Assumption A.1 limits the temporal dependence in $\{R_{t}\}$ under a $\beta $-mixing structure. Assumption A.2 imposes common smoothness restrictions on the functions of interest \citep[see][]{cai2007trending, orbe2005nonparametric, robinson1989nonparametric}, and requires slightly more than four moments of the data. More importantly, unlike \citet{cai2007trending} and \citet{chen2012testing}, we allow for time-varying moments which means that the data does not have to be stationary. This is highly relevant because the stationarity of many macroeconomic variables that are of forecasting interest, such as inflation, is still subject to debate. Notably, these conditions are sufficiently general to accommodate nonlinear
locally stationary processes as defined in \citet{dahlhaus2019towards} and \citet{vogt2012nonparametric}, which include time-varying parameter autoregressive processes. Assumption A.3 allows for conditional heteroscedasticity of an unknown form but rules out potential serial correlation in the forecast errors. This condition is reasonable if we expect forecasters to have included all the information known to them at time $t$ in making their forecasts. When more forecasts are used in the combination, the condition is more likely to hold. The m.d.s. assumption greatly simplifies the analysis with high-dimensional data, but we also consider relaxing the assumption to allow for serial correlation below. Lastly, A.4 is a standard assumption for kernel regressions. We note that commonly used second-order kernels, such as the Epanechnikov, uniform and quartic kernels, satisfy this condition. Furthermore, A.4 implies $\int_{-1}^{1}k(u)du=1$, $\int_{-1}^{1}uk(u)du=0$, and $\int_{-1}^{1}u^{2}k(u)du<\infty $. 

We now state the asymptotic properties of $\hat{\beta}(\tau)$, which can be viewed as an extension of Theorem 4 of \citet{cai2007trending} for forecast combination with the reflection method.

\begin{proposition}
\thlabel{consis} If assumptions A.1-4 hold, and $h=O(T^{-1/5})$, then for
all $\tau \in \lbrack 0,1]$, we have 
\begin{equation}
\sqrt{Th}\bigg[\hat{\beta}(\tau )-\beta (\tau )-\frac{h^{2}\beta ^{^{\prime
\prime }}(\tau )\mu _{2}}{2}+o_{p}(h^{2})\bigg]\rightarrow ^{d}N(0,2\nu
_{0}M^{-1}(\tau )V(\tau )M^{-1}(\tau )).
\end{equation}%
where $\nu _{0}=\int_{-1}^{1}k^{2}(u)du$, $\mu _{2}=\int u^{2}k(u)du$, and $%
V(\tau)$ is defined in A.2.
\end{proposition}

\thref{consis} shows that $\hat{\beta}(\tau )$ is a consistent estimator of $\beta (\tau)$ and the asymptotic bias depends on the curvature of  $\beta(\tau).$ Numerical analysis with commonly used kernels shows that the asymptotic variance here is much smaller than what we might expect if we had not used the data reflection. For example, with the Epanechinikov kernel, the reflection method can reduce the asymptotic variance by more than 70\%.

As \citet{diebold1988serial} points out, regression-based methods of forecast combination might lead to serially correlated errors. Therefore, we relax the m.d.s. assumption and consider the following alternative.

\noindent \textbf{Assumption A.3* (Serial correlation):} Let {$\varepsilon_t$%
} satisfy: (i) $E( \varepsilon_{t+1}|X_{t}) =0$, and (ii) $\Gamma
_{j}(t/T)=Cov(X_{t}\varepsilon_{t+1},X_{t+j}\varepsilon _{t+j+1}) $, where $%
\Gamma _{j}(\tau)$ is Lipschitz continuous for all $\tau \in [0,1]$.

\begin{proposition}
\thlabel{consis2} If assumptions A.1,A.2,A.3* and A.4 hold, and $%
h=O(T^{-1/5})$, then for all $\tau \in \lbrack 0,1]$, we have 
\begin{equation}
\sqrt{Th}\bigg[\hat{\beta}(\tau )-\beta (\tau )-\frac{h^{2}\beta ^{^{\prime
\prime }}(\tau )\mu _{2}}{2}+o_{p}(h^{2})\bigg]\rightarrow ^{d}N(0,2\nu
_{0}M^{-1}(\tau )\Omega (\tau )M^{-1}(\tau )).  \label{betanorm}
\end{equation}%
where $\Omega (\tau )=\sum_{j=-\infty }^{\infty }\Gamma _{j}(\tau )$ .
\end{proposition}

\thref{consis2} shows that asymptotic normality continues to hold with potential serial correlation, while assumptions A.1 and A.2 guarantee the existence of the long-run variance $\Omega(\tau)$ for each $\tau \in [0,1]$.

Next, we study the optimal choice of the bandwidth, $h$, under the m.d.s. assumption (A.3). The choice of the bandwidth $h$ is generally believed to be more important than the choice of the kernel function $k(\cdot )$ in estimation. A small $h$ tends to reduce the bias in $\hat{\beta}%
_{t} $ at the expense of variance, and vice versa with large $h$. Hence, we
opt for a data-driven method to select the bandwidth by optimizing some
metric of forecast errors. From \thref{consis}, we obtain the mean squared
combined forecast errors (MSCFE) as 
\begin{align}
MSCFE_{t}(h)& =E[(y_{t+1}-X_{t}^{\top }\beta _{t})^{2}]  \notag \\
& =E[\varepsilon _{t+1}^{2}]+E[(\hat{\beta}_{t}-\beta _{t})^{\top
}X_{t}X_{t}^{\top }(\hat{\beta}_{t}-\beta _{t})].
\end{align}

Here, we eliminate the cross product term since $\hat{\beta}_{t} $ only uses information up to time $t$. Subsequently, define the integrated MSCFE as 
\begin{align}
IMSCFE(h)& =\int_{0}^{1}MSCFE_{t}(h)d(t/T)  \notag  \label{IMSCFE} \\
& =\int_{0}^{1}\sigma^2 (\tau )d\tau +\int_{0}^{1}\Tr{\bigg[M(\tau)\bigg\{%
\frac{h^4\mu_2^2}{4} \beta^{''}(\tau)\beta^{''}(\tau)^\top + \frac{2\nu_0
V_{\beta}(\tau)}{Th}\bigg\}\bigg]}d\tau ,
\end{align}%
where $V_{\beta }(\tau )\equiv M^{-1}(\tau )V(\tau )M^{-1}(\tau )$ and label the second term in \eqref{IMSCFE} as $IMSCFE(h)_{L}$. To obtain the optimal bandwidth, we minimize the $IMSCFE(h)$ or equivalently $IMSCFE(h)_{L}$ with respect to $h$, which yields 
\begin{equation}
h^{opt}=T^{-\frac{1}{5}}\bigg(\frac{2\nu _{0}\int \Tr\lbrack {V(\tau
)M^{-1}(\tau )]d\tau }}{\mu _{2}^{2}\int \Tr[{M(\tau) \beta^{''}(\tau)
\beta^{''}(\tau)^\top}]d\tau }\bigg)^{\frac{1}{5}},
\end{equation}%
and hence the optimal convergence rate of the IMSCFE is of the order $O(T^{-4/5})$.

In practice, we can use a leave-one-out cross-validation (CV) to select the bandwidth. Specifically, a data-driven choice of $h$ is obtained by solving the following problem, 
\begin{equation}
\label{CVselect}
\hat{h}_{CV}=\underset{c_{1}T^{-1/5}\leq h\leq c_{2}T^{-1/5}}{\argmin}CV(h),
\end{equation}%
where $CV(h)=T^{-1}\sum_{s=1}^{T}(y_{s+1}-X_{s}^{\top }\hat{\beta}_{s})^{2}$%
, and $c_{1}$ and $c_{2}$ are suitable constants. To formalize the
optimality of the bandwidth selected by CV, we require
stronger moment assumptions.

\noindent \textbf{Assumption A.5 (CV moment condition):} Assume $\sup_{1\leq
t\leq T}E\Vert R_{t}\Vert ^{12}<\infty $.

Assumption A.5 is imposed to facilitate technical derivation. \cite{hardle1985optimal} and \cite{xia2002asymptotic} impose similar moment conditions.

\begin{theorem}
\thlabel{CVunif} Suppose assumptions A.1-5 are satisfied. As $T\rightarrow
\infty $ 
\begin{equation}
\hat{h}_{CV}/h^{opt}\rightarrow^{p} 1.
\end{equation}
\end{theorem}

The estimated bandwidth derived from minimizing CV is asymptotically optimal in the sense that it minimizes the IMSCFE. Unlike plug-in methods which are directly based on the theoretical optimal bandwidth $h^{opt}$, the CV method does not require the preliminary estimation of asymptotic bias or variance. It can be implemented automatically and we expect that it would have reasonable finite sample performance.

\section{High-Dimensional Forecast Combination}

\label{highdimforecast}

When the dimension of forecasts is large, the local linear estimation does not work well. Moreover in practice, some individual forecasts may or may not be important or relevant. Therefore, we combine the local linear estimation with regularization for model selection and estimation of the combination weights in a high-dimensional context.

Consider 
\begin{equation*}
y_{t+1}=\omega _{0t}+f_{t}^{\top }\omega _{1t}+\varepsilon _{t+1},\qquad
t=1,...,T,
\end{equation*}%
where the $f_{t}=\left( f_{t1,}f_{t2},...,f_{tp_{T}}\right) ^{\top }$\text{
is }$p_{T}\times 1$ and $p_{T}$, the number of candidate forecasts, can be larger than the sample size $T.$ We assume that there exists $d\ll T$ and $1\leq d<p_{T}$\text{ such that }$%
\omega _{1t,j}\neq 0$ for $1\leq j\leq d$ and $\omega _{1t,j}=0$
for $d<j\leq p_{T}.$ In other words, there are $d$ relevant forecasts. Moreover, the dimension of important individual forecasts $d$ may diverge with $T$. A natural
estimator of forecast weights would be the local linear estimator with the
Lasso penalty: 
\begin{equation}
\tilde{\gamma}_{t}=\underset{\gamma _{t}\in \mathbb{R}^{2(p_{T}+1)}}{\arg
\min }T^{-1}\sum_{s=t-\floor* {Th},s\neq t}^{t+\floor* {Th}}k_{st}%
\left[ y_{s+1}-\alpha _{0t}^{\top }X_{s}-\alpha _{1t}^{\top }\left( \frac{s-t%
}{T}\right) X_{s}\right] ^{2}+\lambda _{1}\left\vert \alpha _{0t}\right\vert
+\lambda _{2}\left\vert h\alpha _{1t}\right\vert ,  \label{prelim}
\end{equation}%
where $\lambda _{1}$ and $\lambda _{2}$ are two tuning parameters, and $%
\gamma _{t}=(\alpha _{0t}^{\top },\alpha _{1t}^{\top })^{\top }$. 

Then, the local linear estimator for $\beta _{t}$ is given by 
\begin{equation}
\tilde{\beta}_{t}=(e_{1}^{\top }\otimes I_{(p_{T}+1)})\tilde{\gamma}_{t}.
\label{betat}
\end{equation}%
Define $\tilde{\Upsilon}=\left( \tilde{\gamma}_{1},\tilde{\gamma}_{2},...,%
\tilde{\gamma}_{T}\right) ^{\top }$ and $\tilde{B}=\left( \tilde{\beta}_{1},%
\tilde{\beta}_{2},...,\tilde{\beta}_{T}\right) ^{\top }$. It is easy to
verify that 
\begin{align}
\tilde{\Upsilon}=& \underset{\Upsilon \in \mathbb{R}^{T\times 2(p_{T}+1)}}{%
\arg \min }T^{-1}\sum_{t=1}^{T}\sum_{s=t-\floor* {Th},s\neq t}^{t+\lfloor
Th\rfloor }k_{st}\left[ y_{s+1}-\alpha _{0t}^{\top }X_{s}-\alpha _{1t}^{\top
}\left( \frac{s-t}{T}\right) X_{s}\right] ^{2}  \notag \\
& +\lambda _{1}\left\vert \alpha _{0t}\right\vert +\lambda _{2}\left\vert
h\alpha _{1t}\right\vert .  \label{fslasso}
\end{align}%
The Lasso-based local linear estimator is estimation consistent (see \thref{fsprop}) but requires very strong assumptions for selection consistency. Instead, following \cite{li2015model}, we minimize \eqref{fslasso} to obtain preliminary estimates for use in a second-stage penalized optimization with the (group) SCAD penalty proposed by \cite{fan2001variable}. In particular, we use the initial
estimates from $\tilde{B}$ to solve 
\begin{eqnarray}
\hat{\Upsilon}^{h} &=&\underset{\Upsilon \in \mathbb{R}^{T\times 2(p_{T}+1)}}%
{\arg \min }T^{-1}\sum_{t=1}^{T}\sum_{s=t-\floor* {Th},s\neq t}^{t+\floor* {Th}}k_{st}\left[ y_{s+1}-\alpha _{0t}^{\top }X_{s}-\alpha _{1t}^{\top
}\left( \frac{s-t}{T}\right) X_{s}\right] ^{2}  \label{llscad} \\
&&+\sum_{j=1}^{p_{T}+1}p_{\lambda _{3}}^{^{\prime }}\left( \left\Vert \tilde{%
B}_{j}\right\Vert \right) \left\Vert \alpha _{0,j}\right\Vert
+\sum_{j=1}^{p_{T}+1}p_{\lambda _{4}}^{^{\prime }}\left( \tilde{D}%
_{j}\right) \left\Vert h\alpha _{1,j}\right\Vert ,  \notag
\end{eqnarray}%
where $\lambda _{3}$ and $\lambda _{4}$ are two tuning parameters, $\alpha
_{i}=\left( \alpha _{i1},\alpha _{i2},...,\alpha _{iT}\right) ^{\top }$ for $%
i=0,1,$ and $\alpha _{i,j}$ is the $j^{th}$ column of $\alpha _{i},$ $\tilde{%
B}_{j}$ is the $j^{th}$ column of the first-stage Lasso-based local linear
estimator $\tilde{B},$ and 
\begin{equation*}
\tilde{D}_{j}=\left\{ \sum_{t=1}^{T}\left[ \tilde{\beta}_{t,j}-\frac{1}{T}%
\sum_{s=1}^{T}\tilde{\beta}_{s,j}\right] ^{2}\right\} ^{1/2},
\end{equation*}%
which measures the smoothness of the LASSO-based local linear estimator $%
\tilde{B}.$ Moreover, $p_{\lambda }^{'}\left( \cdot \right) $ is the
derivative of the SCAD penalty function with regularization parameter $%
\lambda $ defined by 
\begin{equation*}
p_{\lambda }^{' }\left( x\right) =\lambda \left[ 1\left( x\leq \lambda
\right) +\frac{\left( a\lambda -x\right) _{+}}{\left( a-1\right) \lambda }%
1\left( x>\lambda \right) \right]
\end{equation*}%
and $a=3.7$ as suggested in \cite{fan2001variable}. Instead of the SCAD
penalty itself, we use a local linear approximation of the penalty to
overcome difficulties due to non-convexity of the SCAD penalty %
\citep{zou2008one, fan2014strong}. Then the local linear estimator for $%
\beta _{t}$ with the group SCAD penalty is%
\begin{equation}
\hat{\beta}_{t}^{h}=(e_{1}^{\top }\otimes I_{(p_{T}+1)})\hat{\gamma}_{t}^{h},
\label{betah}
\end{equation}%
where $\hat{\gamma}_{t}^{h}$ is the $t^{th}$ row of $\hat{\Upsilon}^{h}.$

\section{Computational Algorithm}

\label{notation}

We approach the estimation of our forecast weights with a two-stage strategy\footnote{Throughout the process, we standardize our data even though the forecasts and the variable of interest are expected to share the same scale because we find that it helps with the stability of the algorithm and is standard
practice in the Lasso literature. This involves centering the data and dividing by its standard deviation using the whole sample.}. In the first
stage, we solve \eqref{prelim} to obtain preliminary coefficient estimates. Subsequently, we use these to initialize the group coordinate descent algorithm (Yuan and Lin, 2006; Wei et al., 2011) in order to solve the penalized regression with the group SCAD penalties. The first-stage is a standard problem that can be solved efficiently for each $t$ by accessible statistical programs, such as {\fontfamily{lmtt}\selectfont glmnet} in {\fontfamily{lmtt}\selectfont R}, so we focus our attention to the second-stage problem with the group SCAD penalties.

For convenience, we rewrite the first term of \eqref{llscad} in matrix
notation as 
\begin{equation}
\mathcal{L}^\diamond(\alpha_0, \alpha_1) = T^{-1} \bigg(\overline{Y}%
-\sum_{j=1}^{p_{T}+1}\Xi _{j}\alpha _{0,j}-\sum_{j=1}^{p_{T}+1}\Xi
_{j+p_{T}+1}\alpha _{1,j}\bigg)^{\top }\overline{K}\bigg(\overline{Y}%
-\sum_{j=1}^{p_{T}+1}\Xi _{j}\alpha _{0,j}-\sum_{j=1}^{p_{T}+1}\Xi
_{j+p_{T}+1}\alpha _{1,j}\bigg)  \label{llmatrix}
\end{equation}
where $\alpha_{i,j}$ is as previously, the $T \times 1$ $j^{th}$ column of $%
\alpha_i$ for $i=0,1$. $\overline{Y}$ is a $2\floor*{Th}T\times 1$ vector
obtained by stacking $Y_{t}$ for $t=1,\ldots ,T$ which is in turn a $2\floor*%
{Th} \times 1$ vector obtained by stacking $y_{s+1}$ for $s$ from $t-\floor*{%
Th}$ to $t+\floor*{Th}$ excluding $t$. $\overline{K}$ is a $2\floor*{Th}%
T\times 2\floor*{Th}T$ block-diagonalization of $\{K_{t}\}_{t}^{T}$, where $%
K_t$ is a diagonal matrix with diagonal elements corresponding to $k_{t-%
\floor*{Th},t},\ldots,k_{t+\floor*{Th},t}$ excluding $k_{t,t}$. $\Xi_{i}$ is
a selection matrix such that 
\begin{equation}
\underset{2\floor*{Th}T \times T}{\Xi_{i}}= 
\begin{bmatrix}
e_{1}^{\top }\otimes Q_{1}e_{i,2(p_{T}+1)} \\ 
\vdots \\ 
e_{T}^{\top }\otimes Q_{T}e_{i,2(p_{T}+1)}%
\end{bmatrix}
\label{Xi}
\end{equation}
where $Q_t$ is obtained by vertically stacking $(X_s^\top, X_s^\top(\frac{s-t%
}{T}))$ for $s$ from $t-\floor*{Th}$ to $t+\floor*{Th}$ excluding $t$, and $%
e_t$ and $e_{i,2(p_T+1)}$ are $T \times 1$ and $2(p_T+1) \times 1$ unit
vectors with unity in the $t^{th}$ and $i^{th}$ coordinates respectively.

To use the group coordinate descent algorithm as in \cite{wei2011variable}
and \cite{yuan2006model}, we need to orthogonalize the matrix $\Xi
_{i}^{\top } \overline{K} \Xi _{i}$. This can be achieved by
post-multiplying $\Xi_i$ with the inverse of the Cholesky decomposition of $%
\Xi _{i}^{\top } \overline{K} \Xi _{i}$ (let it be $A_i$), so that $(\Xi_i
A_{i})^{\top } \overline{K} (\Xi_i A_{i})=I$. Hence, we proceed with the
assumption that $\Xi _{i}$ has been orthogonalized.

Then, it can be shown that 
\begin{equation*}
\alpha _{0,i}=\bigg(1-\frac{\tau _{i}}{\Vert \tilde{S}_{i}\Vert }\bigg)_{+}%
\tilde{S}_{i}\quad \text{and}\quad \alpha _{1,i}=\bigg(1-\frac{h\tau
_{i}^{\ast }}{\Vert \tilde{S}_{i}^{\ast }\Vert }\bigg)_{+}\tilde{S}%
_{i}^{\ast }
\end{equation*}
where $\tau _{i}=p_{\lambda _{3}}^{' }\left( \left\Vert \tilde{B}
_{i}\right\Vert \right) ,\tau _{i}^{\ast }=p_{\lambda 4}^{' }\left( 
\tilde{D}_{i}\right)$, and $\tilde{S}_{i}=\Xi_{i}^{\top } \overline{K}(%
\overline{Y}-\sum_{j \neq i} \Xi_{j} \alpha_{0j}-\sum_{j=1}^{p_{T}+1}
\Xi_{j+p_{T}+1} \alpha_{1j})$ and $\tilde{S}_{i}^{\ast
}=\Xi_{i+p_{T}+1}^{\top }\overline{K}(\overline{Y}-\sum_{j=1}^{p_{T}+1}
\Xi_{j}\alpha _{0,j}-\sum_{j\neq i} \Xi_{j+p_{T}+1}\alpha _{1,j})$. These
relationships allow us to iteratively compute the parameters through the
following algorithm.

\smallskip

Step 1. Initialize with estimates from the first-stage lasso.
In other words, $\alpha_{0,i}^{(0)} = \tilde{\alpha}_{0,i}$ or $%
\alpha_{1,i}^{(0)} = \tilde{\alpha}_{1,i}$. Define $r_0^{(0)} = r_1^{(0)} =  \overline{Y}$.

\smallskip

Step 2. Construct $\tilde{S}_i^{(k+1)} = \Xi_i^\top \overline{K%
} r_0^{(k)} + \alpha_{0,i}^{(k)}$ or $\tilde{S}_i^{*(k+1)} = \Xi_{i+d+1}^\top 
\overline{K} r_1^{(k)} + \alpha_{1,i}^{(k)}$.

\smallskip

Step 3. Use the relation $\alpha_{0,i}^{(k+1)} = (1 - \frac{%
\tau_i^{(k)}}{\| \tilde{S}_i^{(k)} \|})_+ \tilde{S}^{(k+1)}_i$ or $%
\alpha_{1,i}^{(k+1)} = (1 - \frac{h\tau_i^{*(k)}}{\| \tilde{S}_i^{*(k)} \|}%
)_+ \tilde{S}_i^{*(k+1)}$ to update, and use these to form $\tau_i^{(k+1)}$
and $\tau_i^{*(k+1)}$.

\smallskip

Step 4. Construct a new $r$ with $r_0^{(k+1)} = r_0^{(k)} -
\Xi_i^\top \overline{K} ( \alpha_{0,i}^{(k+1)} - \alpha_{0,i}^{(k)})$ or $%
r_1^{(k+1)} = r_1^{(k)} - \Xi_{i+d+1}^\top \overline{K} ( \alpha_{1,i}^{(k+1)} -
\alpha_{1,i}^{(k)})$.

\smallskip

Step 5. Repeat for all the coefficient vectors until a
reasonable tolerance is achieved. We use $1\times10^{-3}$ in our simulations
and applications.

\smallskip

Step 6. Recover the original parameters by applying the
reverse transformation for orthogonalization and standardization.

To implement the local linear estimation with the group SCAD penalty, we need to choose tuning parameters $\lambda _{j}$ and $h$. First, for the preliminary estimates, the tuning parameters $\lambda_{1}$ and $\lambda _{2}$ are obtained via K-fold CV. Given the potentially high computational costs involved in our two-step procedure, CV is attractive because it is readily accessible as the default option of many Lasso-type algorithms in statistical programs\footnote{For example, {\fontfamily{lmtt}\selectfont glmnet} and {\fontfamily{lmtt}\selectfont lars} in {\fontfamily{lmtt}\selectfont R}}. Furthermore, in a high-dimensional setting, \cite{homrighausen2017risk} have shown that the CV estimate for Lasso is risk consistent for the oracle tuning parameter. This result might be stronger than what we require in section \ref{high.asymp} since we do not expect the preliminary estimator to be selection consistent for our asymptotic results. 

Following \cite{li2015model}, we set the bandwidth as $h=C[\log(p_{T}+1)/T]^{1/5}$ to minimize the computational burden from additional tuning. Simulation studies show that the preliminary
estimation is not very sensitive to bandwidth selection. For the second-stage estimation, the tuning parameters $\lambda_{3}$ and $\lambda_{4}$ are selected via a modified version of BIC: 
\begin{equation*}
BIC=\log (SSR)+C_{T}\left( l\log \floor*{Th}/\floor*{Th}\right) ,
\end{equation*}%
where $SSR=T^{-1}\sum_{t=1}^{T}(y_{t+1}-x_{t}^{\top }\hat{\beta}%
_{t}^{h})^{2},$ $l$ is the number of significant forecasts (i.e. maximum possible value of $l$ is $d$) obtained given a pair of candidate tuning parameters, and $\lfloor Th\rfloor $ is the effective sample size for estimating time-varying parameters. $C_{T}=\log p_{T}$ is selected to guarantee the consistency of BIC in a high-dimensional regression, following \cite{wang2009shrinkagetuning}. Here, BIC is a popular tuning approach for variable selection problems with SCAD penalties. For example, \cite{cai2015functional} use the BIC for tuning the SCAD penalty in penalized functional coefficient models, although we remark that the tuning parameters we seek, $\lambda_3$ and $\lambda_4$, are less restrictive than theirs\footnote{In particular, we require $\lambda_3 \propto \lambda_4 = o(T^{1/2})$ in HD.3(ii), while $\lambda = o(T^{1/10})$ in \cite{cai2015functional}}.

\section{Asymptotic Analysis with High Dimension}

\label{high.asymp}

To study the asymptotic properties of $\tilde{\beta}_{t}$ and $\hat{\beta}%
_{t}^{h},$ we impose the following additional assumptions using the notation established in the previous section. 

\bigskip

\noindent \textbf{Assumption HD.1 (Moment conditions):} \setlist{nolistsep}

\begin{enumerate}[label=(\roman*), noitemsep]

\item Let $X_t^{o}$ contain the first $d$ relevant forecasts. Define $M^{o}(t/T) = E(X_t^{o} X_t^{o^\top})$, $\sigma^2(t/T) = E(\varepsilon_t^2)$ and $%
V^{o}(t/T) = E(X_t^{o} X_t^{o^\top} \varepsilon_{t+1}^2)$, where $M^{o}(\tau)$, $\sigma^2(\tau)$ and $V^{o}(\tau)$ are Lipschitz continuous for all $\tau \in
[0,1]$, and $M^{o}(\tau)$ is positive definite.

\item Assume that uniformly in $t$, $\max_{1\leq j\leq
(p_{T}+1)}E[|\varepsilon _{t+1}X_{tj}|^{\psi }]<\infty $ for some sufficiently large $\psi >2+%
\frac{\delta _{1}}{1-\delta _{2}}+\delta $, $\delta >0,$ where $\delta _{1}$
and $\delta _{2}$ are defined in assumption HD.3 below, and $X_{sj}$ refers to the $%
j^{th}$ element of $X_{s}$. 
\end{enumerate}

\bigskip

\noindent \textbf{Assumption HD.2 (Restricted eigenvalues):} Define the set $%
S$ as 
\begin{align*}
S=\bigg\{ &v=(v_{11},\ldots,v_{1p_T+1},v_{21},\ldots,v_{2p_T+1})^\top \in 
\mathbb{R}^{2(p_T+1)}: \|v\|=1, \\
&\sum_{m=1}^{p_T+1}(|v_{1m}|+|v_{2m}|)\leq 2(1+\delta)\sum_{m=1}^{d}
(|v_{1m}|+|v_{2m}|) \bigg\},
\end{align*}
for some $\delta>0$. Then there exists positive constants $%
0<\rho_1\leq\rho_2<\infty$ with probability approaching one such that, 
\begin{equation*}
\rho_1 \leq T^{-1} \inf_{\tau \in [0.1]} \inf_{v \in S} v^\top
Q_\tau^{h\top} K_\tau Q_\tau^h v \leq T^{-1} \sup_{\tau \in [0.1]} \sup_{v
\in S} v^\top Q_\tau^{h\top} K_\tau Q_\tau^h v \leq \rho_2
\end{equation*}
where $Q^h_\tau= Q_\tau H$, $H=diag\{I_{p_T+1\times p_T+1}, \mathbf{h}%
^{-1}\} $, and $\mathbf{h}^{-1}$ is a $p_T+1\times p_T+1$ diagonal matrix
with diagonal elements $1/h$.\footnote{%
Effectively, $Q_t^h$ is the matrix constructed by vertically stacking $%
(X_s^\top,X_s^\top(\frac{t-s}{Th}))$ from $s = t-\floor*{Th}$ to $t+\floor*{%
Th}$ excluding $t$.} $Q_\tau$ and $K_\tau$ are defined in the discussion of %
\eqref{llmatrix} in Section \ref{notation}.

\noindent \textbf{Assumption HD.3 (Rates and tuning parameters):} %
\setlist{nolistsep}

\begin{enumerate}[label=(\roman*), noitemsep]

\item Let $p_{T}=c_{1}T^{\delta _{1}}$, $h=c_{2}T^{-\delta _{2}}$, $\lambda
_{1}\varpropto \lambda _{2}$, where $0\leq \delta _{1}<\infty $, $0<\delta
_{2}<1.$ The bandwidth and the tuning parameter $\lambda _{1}$ satisfy $%
d h^{2}\lambda _{1}^{-1}+d h^{-2}\lambda_{1}^{2}+ d\lambda _{1}^{1/2}+\left( \log h^{-1}/Th\right) ^{1/2}\lambda _{1}^{-1}\rightarrow 0$%
.

\item Let $dh^2 \varpropto (Th)^{-1/2}$, $\lambda _{3}\varpropto \lambda
_{4}$, $\lambda_3 = o(T^{1/2})$, and $h^{-1/2}[(\log h^{-1})^{1/2}+
d^{1/2} + \lambda_1 h^{1/2}\sqrt{Td}]\lambda _{3}^{-1} \rightarrow 0.$

\item With probability approaching one, there exists a positive constant $%
b_\diamond$ such that 
\begin{equation*}
\min_{1 \leq j \leq d} \|B_j\| \geq b_\diamond T^{1/2},\text{ and } \min_{1
\leq j \leq d_1} D_{j} \geq b_\diamond T^{1/2}.
\end{equation*}
\end{enumerate}

\bigskip

Assumption HD.1(i) is the high-dimensional counterpart of (iii) in assumption 2. HD.1(ii) is a moment condition similar to that in \cite{li2015model}, while HD.2 is a generalization of the restricted eigenvalue conditions in \cite{bickel2009simultaneous}. The regularity conditions in HD.3(i) allows the number of forecasts to increase at a polynomial rate and imposes restrictions on the penalty parameters, $\lambda _{1}$ and $\lambda _{2}$ and the bandwidth, $h$, which are required for showing \thref{fsprop}. On the other hand, HD.3(ii) allows $\lambda _{3}\rightarrow \infty $, albeit at a slower rate than $\sqrt{T}$, and is used to show \thref{ssoracle}. Finally, HD.3(iii) requires the coefficients on relevant forecasts to be bounded away from 0, which is used to show the oracle property in \thref{oracleprop}.

We first establish the asymptotic properties of the first-stage estimator $%
\tilde{\beta}_{t}.$

\begin{proposition}
\thlabel{fsprop} If assumptions A.1,A.2(i)-(ii),A.3, A.4, HD.1-2, and HD.3(i) hold, we have 
\begin{equation}
\max_{t}\Vert \tilde{\beta}_{t}-\beta _{t}\Vert \rightarrow^{p} 0
\end{equation}%
as $T\rightarrow \infty .$
\end{proposition}

\thref{fsprop} shows that the local linear estimator with Lasso penalty is estimation consistent but it is not variable selection consistent in the absence of strong "irrepresentability" conditions for the Lasso in a high-dimensional setting \citep[see for e.g.][]{zhang2010nearly}. Therefore, we only use $\tilde{\beta}_{t}$ as the first-stage estimator to figure out the initial weights for use in the group SCAD penalty. To study the selection consistency of
the proposed two-stage method, we define $S=\left\{ j_{1},...,j_{d^{\ast
}}\right\} $ as the index set of an arbitrary model with a total of $0\leq d^{\ast }\leq p_{T}
$ non-zero coefficients (i.e.$X_{tj_{1}},...,X_{tj_{d^{\ast }}}).$ Then we
use $S_{0}=\left\{ 1,...,d\right\} $ to denote the index set of the true model and $\hat{S}%
=\{ j:\Vert \hat{B}_{j}^{h}\Vert >0\} $ to represent
the model selected by the two-stage procedure, where $\hat{B}_{j}^{h}=( \hat{\beta}_{1,j}^{h},...,%
\hat{\beta}_{T,j}^{h}) ^{\top }.$

\begin{theorem}
\thlabel{ssoracle} Assume assumptions A.1, A.2(i)-(ii), A.3, A.4, and HD.1-3
hold. We have
\begin{equation} \label{select}
P\left( \hat{S}=S_{0}\right) \rightarrow 1
\end{equation}
and
\begin{equation} \label{est}
\max_{t} \|\hat{\beta}_t^{h} - \beta_t\| \rightarrow^{p} 0 
\end{equation}
as T$\rightarrow \infty .$
\end{theorem}

\thref{ssoracle} shows that the two-stage method is not only estimation consistent but can also consistently select all relevant individual forecasts. Next, we establish the oracle property. Let $\hat{\beta_t}^{o,h}$ and $\beta_t^{o}$ represent the first $d$ non-zero forecast weights in $\hat{\beta_t}^h$ and $\beta_t$ respectively. The following theorem states that the two-stage group SCAD estimator is asymptotically normal. 

\begin{theorem}
\thlabel{oracleprop} Assume A.1, A.2(i)-(ii), A.3, A.4, and HD.1-3 hold. Then, for all $\tau \in [0,1]$, 
\begin{equation}
\sqrt{Th} A_T \Omega^{o^{-1/2}}(\tau)\bigg\{ \hat{\beta}^{o,h}(\tau) - \beta^{o}(\tau) - \frac{h^2}{2} \mu_2 \beta^{o^{''}}(\tau) + o_p(h^2) \bigg\} \rightarrow^{d} N(0,G)
\end{equation}%
as $T\rightarrow \infty$, where $\Omega^o(\tau) = 2 \nu_0 M^{o^{-1}}(\tau)V^{o}(\tau)M^{o^{-1}}(\tau)$, and $A_T$ is an arbitrary $q \times d$ matrix\footnote{Similar to \cite{fan2004nonconcave}, we consider the asymptotic normality of arbitrary linear combinations of $\hat{\beta}^{o,h}(\tau)$ by pre-multiplying $A_T$ because the dimensions depend on $T$ and might diverge as $T$ goes to infinity.} such that $A_T A_T^\top \rightarrow G$ for a given finite $q$
\end{theorem}

\thref{oracleprop} is related to the oracle property because it is derived by showing that the two-stage SCAD estimator is asymptotically equivalent to the oracle estimator\footnote{The estimator obtained from optimizing the likelihood while having \textit{a priori} knowledge of the relevant estimators and excluding the irrelevant ones.}, which implies that both estimators share the same asymptotic behavior. By further establishing the asymptotic normality of the oracle estimator, we can thus translate the property to our estimator and adopt similar tools for statistical inference. We note that in the special case where $d$ does not depend on $T$, the oracle estimator will be identical to the low-dimensional estimator studied in section \ref{sect.asym}, and a direct application of \thref{consis} is possible.       

\section{Monte Carlo Simulation}

In this section, we contrast the out-of-sample forecasting performance of
the nonparametric estimator with that of common forecast combination
techniques, which include both static and time-varying approaches. We first look at the case with only two forecasts. Then we study a high-dimensional setting to evaluate the
finite-sample oracle properties of the proposed estimator with the group
SCAD penalty.

\subsection{Forecast combination with low-dimensional data}

\label{lowsim}

To start off, we consider the following time-varying coefficient model: 
\begin{equation*}
y_{t+1} = \omega_{0t} + \omega_{1t} f_{1,t} + \beta_{2t} f_{2,t} + u_{t+1}
\end{equation*}
where $\omega_{0t} = \exp(-3 + 2.5 \tau)$, $\omega_{1t} = 0.5(1.5 
\tau - 0.8)^3 + 0.5$, and $\omega_{2t} = 0.2\sin(4 \tau) + 0.4$%
, for $\tau = t/T$, while the forecasts evolve according to: 
\begin{align*}
&f_{1,t} = 0.5 + 0.8y_t + e_{1,t} \\
&f_{2,t} = 0.5 + 0.3\sin\bigg(2\tau + 0.25\bigg)y_t + e_{2,t}.
\end{align*}
Lastly, $u_{t+1}$, $e_{1,t}$ and $e_{2,t}$ are normally distributed with
mean 0 and unit variance.

We study this model specification because it can be endowed with an economic
interpretation consistent with the time-varying common factor framework for
forecast combinations in \cite{elliott2005optimal}. Specifically, it can be
shown that this model reduces to a time-varying parameter AR(1) process with
heteroskedastic errors, and hence in this case, the factor is observable and
completely captured by the first lag of $y$.

The experiment is conducted with three samples: $T \in \{200,300,500\}$ and
an out-of-sample period of 50 time points in excess of \textit{T}. For each
case, we require a holdout or burn-in period of $2 \times T$ for the process
to stabilize and to provide points for bandwidth CV.

In addition, we compare the proposed nonparametric estimator to alternative
approaches that can be classified as either static or adaptive. For the
static case, we use $T$ points to estimate the weights, while we employ an
expanding window ($T + k$) for the adaptive estimators as $k$ increases. The
adaptive estimators suffer from the same issue as that of the nonparametric
estimator in that $y_{t+1}$ is not available at time $t$ for estimation. Hence, we estimate $\hat{\beta}_{t-1}$ and use it to approximate $\hat{\beta}_t$ to form a forecast of $y_{t+1}$. This reflects the practice of professional forecasters and is justified by our assumption that $\beta_t$ is smooth.

Finally, we evaluate the performance of all the models by calculating the
average squared combined forecast error (ASCFE), which is defined as the sum
of the squared deviations of the computed $\hat{y_t}$ from $y_t$ obtained
from the out-of-sample period, i.e. $ASCFE = 1/50\sum_{t=T+2}^{T+51}(y_t-%
\hat{y}_t)^2$. The experiment is repeated 500 times to obtain the mean and
standard deviations of the ASCFE.

\subsubsection{Competing forecast combination methods}

\label{competing} \smallskip \noindent \textbf{Nonparametric estimation}
\smallskip

We consider the local linear estimator with data reflection as established in section \ref{sect.asym}. For bandwidth selection, we adopt the CV method introduced in \eqref{CVselect}. For ease of presentation, we label the nonparametric data
reflection method "NPRf". In all subsequent simulations and empirical
applications, we use the Epanechnikov kernel, i.e. $K(u)=0.75(1-u^{2})_{+}$.

\bigskip \noindent \textbf{Bates and Granger (1969)} \smallskip

A commonly used time-varying combination scheme, originally
suggested by \cite{bates1969combination}, is an adaptive updating method which
assigns a higher weight to forecasts that perform comparatively well in the
recent past. In particular, we consider an expanding window version of their
estimator such that the combination weight of forecast $1$ is given by 
\begin{equation*}
\omega _{1t}^{BG}=\frac{\hat{e}_{1,t}^{-1}}{\hat{e}_{1,t}^{-1}+\hat{e}_{2,t}^{-1}}
\end{equation*}%
where $\hat{e}_{i,t}=\frac{1}{t}\sum_{l=1}^{t}(y_{l+1}-f_{i,l})^{2}$. This is
defined analogously for $f_{2}$. Note that there is no intercept in the
model, and the weights sum to 1. We term this method "BG".

\bigskip \noindent \textbf{Least squares regression} \smallskip

We consider three static combination schemes due to \cite{granger1984improved}. They are: 
\begin{align*}
& y_{t+1}=\omega _{0}+\omega _{1}f_{1,t}+\omega _{2}f_{2,t}+z_{t}, \\
& y_{t+1}=\omega _{1}f_{1,t}+\omega _{2}f_{2,t}+z_{t}, \\
& y_{t+1}=\omega _{1}f_{1,t}+\omega _{2}f_{2,t}+z_{t},\text{where }\omega
_{1}+\omega _{2}=1.
\end{align*}
These models are respectively labeled as "GRregconst" for regression with a
constant, "GRreg" for regression without the intercept term, and "GRconstr"
for constrained regression. It is known that despite biased
forecasts, "GRregconst" performs favorably in terms of MSCFE because $\omega _{0}$ is able to capture the bias \cite[see][]{timmermann2006forecast}. Since we have introduced time-variation in our
experiment, we also consider adaptive versions of the models above, whose
weights are estimated ex-ante. These models are named "TVGRregconst",
"TVGRreg", and "TVGRregconstr" accordingly.

\bigskip \noindent \textbf{Equal weights} \smallskip

Lastly, we include the simple strategy of assigning equal weights
to the forecasts, which in this case is expressed by: $\hat{y}_{t+1} =
0.5f_{1,t} + 0.5f_{2,t}$. This is intended to assess whether the proposed
estimators suffer from the "forecast combination puzzle", which refers to
the commonly observed empirical fact that weights derived from a simple
arithmetic mean often outperform theoretical optimal weights based on
sophisticated estimation. We label this "EQ".

\subsubsection{Simulation results}

Table \ref{tab:table1} reports the mean and standard deviation of the ASCFE.
We can see that our estimator outperforms alternative methods by achieving
both the lowest ASCFE, and the smallest variance, at all sample sizes.

\begin{table}[tbp]
\centering
\caption{Simulation results with low-dimensional data (2 forecasts).}
\setlength{\tabcolsep}{20pt} 
\renewcommand{\arraystretch}{1} 
\begin{threeparttable}
\begin{tabular}{lccc}
\toprule & \multicolumn{3}{c}{Sample size (T)} \\ 
Estimation & 200 & 300 & 500 \\ 
\midrule \textit{Adaptive} &  &  &  \\ 
NPRf & 1.06 & 1.06 & 1.03 \\ 
& (0.22) & (0.21) & (0.21) \\ 
BG & 1.19 & 1.22 & 1.20 \\ 
& (0.24) & (0.23) & (0.24) \\ 
TVGRregconst & 1.08 & 1.09 & 1.07 \\ 
& (0.22) & (0.21) & (0.22) \\ 
TVGRreg & 1.13 & 1.14 & 1.12 \\ 
& (0.23) & (0.22) & (0.23) \\ 
TVGRregconstr & 1.17 & 1.19 & 1.17 \\ 
& (0.23) & (0.22) & (0.23) \\ 
\textit{Static} &  &  &  \\ 
GRregconst & 1.09 & 1.11 & 1.08 \\ 
& (0.24) & (0.22) & (0.22) \\ 
GRreg & 1.14 & 1.15 & 1.13 \\ 
& (0.23) & (0.22) & (0.23) \\ 
GRregconstr & 1.18 & 1.20 & 1.18 \\ 
& (0.24) & (0.23) & (0.23) \\ 
EQ & 1.29 & 1.34 & 1.34 \\ 
& (0.26) & (0.25) & (0.26) \\ 
\midrule Best & NPRf & NPRf & NPRf \\ 
\bottomrule &  &  & 
\end{tabular}%
\begin{tablenotes}[flushleft]
\small
\item Notes: (1) Mean and standard deviation in parentheses of ASCFE from 500 iterations. (2) NPRf: nonparametric estimator with data reflection, BG: \cite{bates1969combination} adaptive etimator, GRregconst: OLS regression with intercept, GRreg: OLS regression without intercept, GRregconstr: OLS regression with sum of coefficients constrained to unity, TV: time-varying weights, EQ: equal weights. 
\end{tablenotes}
\end{threeparttable}
\label{tab:table1}
\end{table}

In addition, two observations are salient. First, adaptive estimators
perform better than their static counterparts, which is expected because the
true DGP contains many time-varying parameters. In
addition, within the two categories, approaches that allow for an
intercept (TVGRregconst and GRregconst) achieve relatively low ASCFEs which
is likely attributable to their ability in capturing the (time-varying) bias
in the true model. By this logic, it is reasonable that NPRf performs well because it
accommodates both time-varying biases and weights. The lackluster
performance of EQ is also consistent with this reasoning as it permits
neither.

For robustness, we look at 3 more cases in the online appendix. The cases are
modifications of the original DGP: (1) bias is removed
and time-varying weights replaced with constant weights; (2) a constant bias
is permitted and time-varying weights replaced with constant weights; and
(3) bias is removed but weights remain time-varying. The key message from
this exercise is that NPRf performs the best when forecast weights are
indeed changing over time. Regression-based methods appear to perform
slightly better otherwise, although NPRf is a close runner-up.

\subsection{Forecast combination with high-dimensional data}

\label{highsim}

Now we consider the performance of NPRf with the group SCAD penalty in
situations where the number of forecasts may be larger than the effective
sample size. The DGP of section \ref{lowsim} is extended by
considering $J$ additional forecasts $\{f_{j,t}\}_{j=1}^{J}$, for each $t$,
where $J\in \{10,50,100\}$. However, these forecasts are redundant in that $%
\beta _{jt}=0$ for all $j=1,\ldots ,J$ and all $t$. The redundant forecasts
are generated from a joint normal distribution with mean \textbf{0} and
variance-covariance matrix $\Sigma $ such that $cov(f_{j,t},f_{j^{'},t})=2\exp (-|j-j^{^{' }}|)$ and $j,j^{^{'}}=1,\ldots ,J$. To
minimize the computational burden, we consider an out-of-sample period of 10
points for the sample sizes 50,100, and 150, and with 200 Monte Carlo
simulations.

As mentioned in section \ref{notation}, we choose the respective penalty parameters via K-fold CV in the first-stage and BIC in the second. The bandwidth is selected using a rule as in \cite{li2015model}: $h=(\log
(J+3)/T)^{0.2}$. The initial values used in the
group coordinate descent algorithm are obtained from a first-stage Lasso.

Since NPRf is a local linear estimator, the effective number of regressors
(EReg) is thus $2 \times (J + 3)$ where the addition of 3 refers to the
intercept and the two relevant forecasts. For large $J$, EReg often exceeds
the effective sample size, especially when the bandwidth is small, and the
solution to the weighted least squares problem implicit in NPRf is no longer
unique\footnote{Given that EReg is often larger than the sample size, we do not consider the alternative methods from section \ref{competing} as they would no longer produce reliable estimates}. Hence the goal of this analysis is to assess whether the
penalization scheme can accurately select the relevant forecasts and
estimate their weights.

\begin{table}
  \centering
  \caption{Simulation results for high-dimensional data.}
  \begin{threeparttable}
    \begin{tabular}{rrrr}
    \toprule
    \multicolumn{1}{l}{Sample size} & \multicolumn{3}{c}{\textbf{group SCAD}} \\
\cmidrule{2-4}                  & \multicolumn{1}{c}{ASCFE} & \multicolumn{1}{c}{\% correct for OOS period} & \multicolumn{1}{c}{Relevant included for OOS period} \\
    \midrule
    \multicolumn{1}{l}{10 extra forecasts} &               &               &  \\
    50            & 1.55          & 0.81          & 1.00 \\
    100           & 1.34          & 0.91          & 1.00 \\
    150           & 1.34          & 0.96          & 1.00 \\
    \multicolumn{1}{l}{50 extra forecasts} &               &               &  \\
    50            & 1.53          & 0.73          & 0.98 \\
    100           & 1.40          & 0.90          & 1.00 \\
    150           & 1.25          & 0.91          & 1.00 \\
    \multicolumn{1}{l}{100 extra forecasts} &               &               &  \\
    50            & 1.62          & 0.70          & 0.96 \\
    100           & 1.34          & 0.80          & 1.00 \\
    150           & 1.24          & 0.81          & 1.00 \\
                  &               &               &  \\
    \bottomrule
    \end{tabular}%
    \begin{tablenotes}[flushleft]
    \small
    \item Notes: (1) Mean of ASCFE and share of iterations that selected exactly or included the relevant forecasts for the whole out-of-sample (OOS) period of 10 time points. (2) Results are from 200 iterations.
    \end{tablenotes}
    \end{threeparttable}
  \label{tab:table2}%
\end{table}%

Table \ref{tab:table2} reports the mean ASCFE for the group SCAD strategy,
given different combinations of $T$ and $J$. An immediate observation is
that the ASCFE is comparable to that obtained in the low-dimension case in
table \ref{tab:table1}. Given a fixed $J$, we observe that the ASCFE falls
and the share of iterations that accurately captures the relevant forecasts
increases with larger samples. This provides some evidence that the
estimator is consistent in estimating and selecting important
forecasts.

\section{Empirical Applications}

\label{empirical}

We illustrate the use of our nonparametric estimator by considering two
applications of combined forecasting in macroeconomics and finance.
Specifically, we review the results on forecasting inflation in \cite%
{ang2007macro} by extending their analysis to include recent data.
Subsequently, we follow \cite{rapach2010out} to examine the predictability
of equity returns with forecast combinations using 13 predictors. In the
latter context, since the effective number of variables for the local linear
estimator is two times that of the original, it may still be large relative
to the effective sample size even if it does not exceed it. We show that
applying the group SCAD strategy in this case yields two benefits over
conventional methods: estimation is more precise in terms of smaller errors,
and important equity return predictors can be identified. Predictor selection is of independent interest in the return
predictability literature.

To facilitate statistical comparison between forecast strategies, we conduct the 
\cite{diebold1995comparing} test (DM test) and the 'Reality Check' (RC) test introduced by \cite%
{white2000reality}. The latter can be used to compare forecasts generated from
both nested and non-nested models, which is especially relevant in the
high-dimensional context because of variable selection.

\subsection{Forecasting inflation}

\cite{ang2007macro} studied the performance of various strategies in
forecasting inflation in the US from 1985 to 2002. These approaches can be
classified into four broad categories: ARIMA-type time series models,
Phillips curve-implied forecasts, term structure models, and survey-based
measures. Two major results from their investigation are particularly
striking. First, they find that the median of survey forecasts, in
particular, the Livingstone survey and the Survey of Professional
Forecasters (SPF)\footnote{%
See \cite{ang2007macro} for a detailed description of the surveys. We do not consider the Livingstone survey in our application because it is a biannual survey.}, consistently outperform models from the other three categories. Furthermore, they show that combining forecasts across the different categories using various parsimonious methods, such as least squares regression and equal weights, do not generally yield more accurate forecasts compared to using only survey information. Our goal is thus twofold: we are interested in assessing whether this claim remains valid given an updated evaluation period, and to ascertain the potential gains from using the nonparametric estimator as
opposed to existing methods in combining inflation forecasts.

To achieve this, we use the basic forecasting models considered by \cite{ang2007macro} from each of the four categories: ARMA(1,1), PC1, TS1, and the SPF for a survey-based forecast\footnote{Naming convention and models follow that of \cite{ang2007macro}}. We consider three CPI-based measures of inflation: CPI for all urban consumers (PUNEW), CPI less housing and shelter (PUXHS), and CPI less
food and energy (PUXX). We generate inflation forecasts for a period of 1981Q3 to 2018Q2, and use the last 3 years for our out-of-sample evaluation. We refer interested readers to the online appendix for a full description of the set-up.

For all measures of inflation, we look at the same combination
techniques as detailed in section \ref{competing}. In addition, we consider individual predictions made by the SPF and the ARMA(1,1) model on quarterly inflation which Ang et al. (2007) regarded as a benchmark. We note that the methods TVGRreg and EQ respectively correspond to the "OLS" and "Mean" methods employed in their paper. Table \ref{tab:table4} reports the ASCFEs, the best methods, and the p-values for one-sided DM and RC tests.

\begin{table}[tbp]
  \centering
  \caption{Forecast combination results for inflation.}
  \begin{threeparttable}
    \begin{tabular}{lccc}
    \toprule
    Estimation    & PUNEW         & PUXHS         & PUXX \\
    \midrule
    \textit{Adaptive} &               &               &  \\
    NPRf          & 0.182         & 0.451         & 0.050 \\
    BG            & 0.457         & 0.679         & 0.117 \\
    TVGRregconst  & 0.186         & 0.792         & 0.048 \\
    TVGRreg       & 0.475         & 0.512         & 0.108 \\
    TVGRregconstr & 0.446         & 0.791         & 0.098 \\
    ARMA(1,1)     & 0.206         & 1.897         & 0.061 \\
    \textit{Static} &               &               &  \\
    GRregconst    & 0.186         & 0.929         & 0.048 \\
    GRreg         & 0.477         & 0.518         & 0.110 \\
    GRregconstr   & 0.445         & 0.803         & 0.099 \\
    EQ            & 0.485         & 0.702         & 0.127 \\
    SPF           & 0.467         & 0.842         & 0.441 \\
    \midrule
    Best          & NPRf          & NPRf          & TVGRregconst \\
    \midrule
                  &               &               &  \\
    DM test (p-value) &               &               &  \\
    NPRf          &               &               &  \\
     $<$ EQ         & 0.142         & 0.314         & 0.084 \\
     $<$ SPF        & 0.063         & 0.231         & 0.014 \\
                  &               &               &  \\
    RC test (p-value) &               &               &  \\
    NPRf          &               &               &  \\
     $<$ EQ         & 0.031         & 0.179         & 0.033 \\
     $<$ SPF        & 0.025         & 0.057         & 0.003 \\
    \bottomrule
    \end{tabular}%
    \begin{tablenotes}[flushleft]
\small
\item Notes: (1) ASCFEs for inflation forecasting in the top panel. (2) NPRf: nonparametric estimator with data reflection, BG: \cite{bates1969combination} adaptive etimator, GRregconst: OLS regression with intercept, GRreg: OLS regression without intercept, GRregconstr: OLS regression with sum of coefficients constrained to unity, TV: time-varying weights, EQ: equal weights, SPF: Survey of Professional Forecasters, PUNEW: CPI (all urban) inflation, PUXHS: CPI (less housing and shelter) inflation, PUXX: CPI (less food and energy) inflation. (3) DM test: \protect\cite{diebold1995comparing} test. RC test:
'Reality Check' test \citep{white2000reality}. (4) "$x < y$" indicates a test of null hypothesis of equal predictive ability between $x$ and $y$, with the one-sided alternative of superior predictive ability of $x$ over $y$.
\end{tablenotes}
\end{threeparttable}
\label{tab:table4}
\end{table}

Several comments are in order here. First, the nonparametric method, NPRf, achieves the best performance in two out of three inflation indices, but we note that the various OLS methods are close contenders. For PUXX, both TVGRregconst and GRregconst achieve the lowest ASCFE, but NPRf is a close runner-up. Furthermore, there is statistical evidence that NPRf improves over equal weighting and simple SPF-implied forecasts for PUNEW and PUXX. This finding contrasts with \cite{ang2007macro}, who found that even with forecast combinations, the weight on the SPF forecast dominated all other contributions. Here, we find that other types of forecast can improve inflation forecasting above and beyond using just survey estimates.

It has to be emphasized that an additional merit of employing NPRf in this situation is that we do not have to take a strong and often contentious stand on the stationarity of inflation, since given our formulation of the estimator, it accommodates some extent of non-stationarity including locally stationary processes. Regardless, the evidence presented thus far appear to favor the use of forecast combination for inflation forecasting, which is a notable departure from the conclusions arrived at in \cite{ang2007macro}.

\subsection{Predictability of stock returns}

\label{stock}

Many popular macroeconomic and financial variables do not posses
out-of-sample predictive power in forecasting stock returns even if they may exhibit good in-sample performance \citep{bossaerts1999implementing, welch2008comprehensive}. However, models that do successfully deliver
statistically significant and economically relevant out-of-sample
forecasting gains are often approaches that accommodate model uncertainty and parameter instability \citep{rapach2013forecasting}. This is not particularly surprising given that forecasters do not possess \textit{a priori} information on the "best" models. Furthermore, there is strong evidence that many predictive relationships of stock returns are unstable over time \citep{chen2012testing}. Forecast combination fits well into this framework because it helps to attenuate uncertainty in individual forecasts and can accommodate time variation. In fact, \cite{rapach2010out} reported significant out-of-sample predictive gains in forecasting US stock returns compared to the historical average using combination strategies that overlap with those that we have introduced in section \ref{competing}. Here, we investigate whether our nonparameteric estimator can achieve similar favorable out-of-sample performace.

To do so, we use updated data (till 2018) from \cite{welch2008comprehensive} and construct the predictors and excess return in a similar fashion. The dependent variable is the stock return derived from the S\&P 500 index and is defined as $\Delta P_{t+1}=[\log (P_{t+1}+D_{t})-\log(P_{t})]-R_{t}$, where $P_{t}$ is the index value, $D_{t}$ is the dividends paid on the index, and $R_{t}$ is the 3-month Treasury bill rate. We use 13 predictors following Rapach et al. (2010), whose full description can be found on the online appendix. Our sample period is from 1947Q2 to 2018Q3, and we use the last 5 years, 2013Q4 to 2018Q3, for out-of-sample evaluation. 

As usual, we compare the performance of static and adaptive estimators as introduced in section \ref{competing}. In addition, we look at the
performance of the prevailing historical average defined as $\bar{\Delta} P_{t+1,T-1}=1/(T-1)\sum_{t=1}^{T-1}\Delta P_{t+1}$, which is essentially a moving average with an expanding window, as this was the benchmark in \cite{rapach2010out}, and is supposedly difficult to beat \citep{welch2008comprehensive}.

Table \ref{tab:table5} presents the ASCFEs scaled by a multiplication of 1000. In general, the adaptive methods seem to perform better than the static OLS-based approaches which is consistent with the notion of parameter instability. However, we observe that the NPRf does not perform favorably in this scenario, as it beats neither the historical average nor the equal weighting scheme. One possible reason for this could be due to the high number of predictors ($13 \times 2$) relative to the sample size. We check this intuition by considering the group SCAD (gSCAD) penalized version of NPRf.

Investigating the value of gSCAD for variable selection is particularly relevant given that the number of predictors in the return predictability literature has been on the rise. Furthermore, predictor selection \textit{per se} is of strong independent interest in asset pricing. To further evaluate our penalization method, we compare it with the partially egalitarian approaches introduced by \cite{diebold2019machine}, which seeks to select important variables in a first-stage Lasso and subsequently shrink their weights to the arithmetic mean with a modified Ridge estimator. Their method capitalizes on the frequently reported empirical finding that simple averages of forecasts tend to outperform sophisticated estimation strategies.

\begin{table}[htbp]
  \centering
  \renewcommand{\arraystretch}{0.95} 
  \caption{Forecast combination results for stock returns.}
  \begin{threeparttable}
    \begin{tabular}{lclc}
    \toprule
    Estimation    &               & \multicolumn{2}{c}{ASCFE} \\
    \midrule
    \textit{Adaptive} &               &               &  \\
    NPRf          &               & \multicolumn{2}{c}{3.183} \\
    BG            &               & \multicolumn{2}{c}{2.879} \\
    TVGRregconst  &               & \multicolumn{2}{c}{4.992} \\
    TVGRreg       &               & \multicolumn{2}{c}{4.663} \\
    TVGRregconstr &               & \multicolumn{2}{c}{2.991} \\
    ARMA(1,1)     &               & \multicolumn{2}{c}{2.807} \\
    Historical average &               & \multicolumn{2}{c}{3.034} \\
    \textit{Static} &               &               &  \\
    GRregconst    &               & \multicolumn{2}{c}{12.506} \\
    GRreg         &               & \multicolumn{2}{c}{9.969} \\
    GRregconstr   &               & \multicolumn{2}{c}{3.125} \\
    EQ            &               & \multicolumn{2}{c}{2.888} \\
                  &               & \multicolumn{2}{c}{} \\
    \midrule
    Penalized estimation & ASCFE         & \multicolumn{2}{c}{Selected predictors} \\
\cmidrule{3-4}                  &               & For all periods & Sometimes selected \\
    \midrule
    gSCAD         & 2.650         & \multicolumn{1}{p{10em}}{E/P, SVAR, NTIS, \newline{}TBL, I/K} & \multicolumn{1}{p{10em}}{B/M} \\
    peLasso       & 3.349         & TBL, I/K      &  \\
                  &               &               &  \\
    \midrule
    Best          & gSCAD         & \multicolumn{2}{c}{} \\
                  &               &               &  \\
    \midrule
    DM test (p-value) &               &               &  \\
    gSCAD         &               &               &  \\
     $<$ Hist. avg. &               & \multicolumn{2}{c}{0.065} \\
     $<$ EQ         &               & \multicolumn{2}{c}{0.106} \\
     $<$ peLasso    &               & \multicolumn{2}{c}{0.027} \\
                  &               &               &  \\
    RC test (p-value) &               &               &  \\
    gSCAD         &               &               &  \\
     $<$ Hist. avg. &               & \multicolumn{2}{c}{0.129} \\
     $<$ EQ         &               & \multicolumn{2}{c}{0.232} \\
     $<$ peLasso    &               & \multicolumn{2}{c}{0.093} \\
    \bottomrule
    \end{tabular}%
    \begin{tablenotes}[flushleft]
    \small
    \item Notes: (1) ASCFEs multiplied by 1000. (2) NPRf: nonparametric estimator with data reflection, BG: \cite{bates1969combination} adaptive etimator, GRregconst: OLS regression with intercept, GRreg: OLS regression without intercept, GRregconstr: OLS regression with sum of coefficients constrained to unity, TV: time-varying weights, EQ: equal weights, gSCAD: NPRf with group SCAD penalties, peLasso: partially egalitarian Lasso \citep{diebold2019machine}. (3) DM test: \protect\cite{diebold1995comparing} test. RC test:
'Reality Check' test \citep{white2000reality}. (4) "$x < y$" indicates a test of null hypothesis of equal predictive ability between $x$ and $y$, with the one-sided alternative of superior predictive ability of $x$ over $y$. See online appendix for predictor description.
\end{tablenotes}
  \end{threeparttable}
  \label{tab:table5}%
\end{table}%

The results in table \ref{tab:table5} show that gSCAD is able to achieve the lowest ASCFE and provide some statistical evidence that gSCAD possesses superior predictive ability over the historical average and the competing regularized estimator, peLasso. Although it appears to perform better than EQ, the statistical evidence is borderline. This result is not entirely unexpected given the short out-of-sample evaluation period and arguably unsatisfactory finite-sample properties of tests for out-of-sample predictive ability \citep{elliott2005optimal, chao2001out}. On another note, it is interesting to see that both gSCAD and peLasso agree on the relevance of the contributions from macro-level variables, the treasury rate (TBL) and the investment-to-capital ratio (I/K), in predicting equity returns for the considered sample period. Although gSCAD suggests that at least 4 other predictors are also important. Capturing the predictors that were excluded by peLasso might have contributed to a better performance.     

\section{Concluding Remarks}

In this paper, we have proposed a new approach to forecast combination with time-varying weights by means of a nonparametric local linear estimator with data reflection. The estimator can be adapted to situations where the number of forecasts is large, possibly larger than the sample size, through the use of sparsity and penalization techniques. Theoretically, we have shown that the nonparametric estimator with the group SCAD penalty is consistent in both selecting relevant forecasts and in estimating the weights. Simulations have shown that the nonparametric estimator performs well in situations of parameter instability, and the relative success of our estimator in empirical applications suggests that predictive instability is the norm rather than the exception. Our proposed nonparametric estimator may also be generalized to forecast combinations with volatility and density forecasts, and it may be interesting to explore these avenues in future research given their popularity in applications.

\bibliographystyle{elsarticle-harv}
\bibliography{Chen_Maung}
\processdelayedfloats

 \makeatletter
\efloat@restorefloats \makeatother
\appendix

\section{Mathematical proofs }

\noindent \textbf{Proof of \thref{consis2}.}

\noindent To facilitate the theoretical derivation, we consider an alternative expression for the local linear estimator
\begin{equation}
\hat{\gamma}_t = 
\begin{bmatrix}
S_{0}(t/T) & S_{1}^\top(t/T) \\ 
S_{1}(t/T) & S_{2}(t/T)%
\end{bmatrix}%
^{-1} 
\begin{bmatrix}
R_{0}(t/T) \\ 
R_{1}(t/T)%
\end{bmatrix}
\equiv S^{-1}(t/T) R (t/T)  \label{matsol}
\end{equation}
where 
\begin{equation*}
S_{j}(t/T) \equiv T^{-1} \sum_{ s=t-\floor*{Th}  , s \neq t}^{t+%
\floor*{Th}} X_s X_s^T \left(\frac{s-t}{T}\right)^j k_{st},
\end{equation*}
\begin{equation*}
R_{j}(t/T) \equiv T^{-1} \sum_{ s=t-\floor*{Th}  , s \neq t}^{t+%
\floor*{Th}} X_s \left(\frac{s-t}{T}\right)^j k_{st} y_{s+1}.
\end{equation*}

Note that \thref{consis} is a special case of \thref{consis2}, and hence it
is sufficient to establish the latter. This will be accomplished by
establishing the asymptotic normality for $\hat{\gamma}_t$, which
encompasses the local linear estimator $\hat{\beta}_t$ in the first $(d+1)$
components. To do so, we use the following lemmas, in a similar fashion to %
\citet{cai2007trending}.

First, define the following notation: $r_{j}(t/T) = T^{-1} \sum_{ s=t-\floor*{Th}, s \neq t}^{t+%
\floor*{Th}} X_s (\frac{s-t}{T})^j k_{st} \varepsilon_{s+1}$, $Q_{s,t} = \beta(s/T) - \beta(t/T) - (\frac{s-t}{T})%
\beta^{^{\prime}}(t/T)-\frac{1}{2}(\frac{s-t}{T})^2\beta^{^{\prime\prime}}(t/T)$, $
D_{j}(t/T)=T^{-1} \sum_{s=t-\floor*{Th}, s \neq t}^{t+\floor*{Th}} X_s X_s^\top (\frac{s-t}{T})^j k_{st} Q_{s,t}$, and $B_{j}(t/T) = \frac{1}{2} S_{j+2} (t/T) \beta^{^{\prime\prime}}(t/T)$,
where $S_{j+2}(\cdot)$ is defined in \eqref{matsol}. By straightforward substitution, we get 
\begin{equation}
\hat{\gamma}_t - \gamma_t = S^{-1}(t/T)r(t/T) + S^{-1}(t/T)B(t/T) +
S^{-1}(t/T)D(t/T),  \label{A1}
\end{equation}
where $r(t/T) = (r_0^\top(t/T), r_1^\top(t/T))^\top$, $B(t/T) =
(B_0^\top(t/T), B_1^\top(t/T))^\top,$ and $D(t/T) = (D_0^\top(t/T), \newline
D_1^\top(t/T))^\top$. We shall show that the terms on the right
hand side of \eqref{A1} are asymptotically negligible.

By \thref{lemma1} below, 
\begin{equation*}
S^{-1}(t/T) = 
\begin{bmatrix}
S_0(t/T) & S_1^\top(t/T) \\ 
S_1(t/T) & S_2(t/T)%
\end{bmatrix}%
^{-1} \rightarrow^p 
\begin{bmatrix}
M(t/T) & \mathbf{0} \\ 
\mathbf{0} & h^2\mu_2 M(t/T)%
\end{bmatrix}%
^{-1} \equiv H^{-1} \Sigma^{-1} (t/T) H^{-1},
\end{equation*}
where $\Sigma(t/T) = diag\{M(t/T), \mu_2M(t/T)\}$. In addition, %
\thref{lemma1} implies that $H^{-1}D(t/T) = o_p(h^2)$, and 
\begin{equation*}
HS^{-1}(t/T)B(t/T) \rightarrow^p \Sigma^{-1}(t/T)H^{-1}\frac{h^2}{2}%
\begin{pmatrix}
\mu_2 M(t/T) \beta^{^{\prime\prime}}(t/T) \\ 
0%
\end{pmatrix}
= \frac{h^2}{2} 
\begin{pmatrix}
\mu_2 \beta^{^{\prime\prime}}(t/T) \\ 
0%
\end{pmatrix}%
.
\end{equation*}
Subsequently, from \thref{lemma2} and Chebyshev, we get $H^{-1}r(t/T) =
O_p((Th)^{-1/2})$. Hence, recalling the definition of \eqref{A1}, we
conclude that 
\begin{equation*}
H(\hat{\gamma}(t/T) - \gamma(t/T)) - \frac{h^2}{2}%
\begin{pmatrix}
\mu_2 \beta^{^{\prime\prime}}(t/T) \\ 
0%
\end{pmatrix}
= o_p(h^2) + O_p((Th)^{-1/2}).
\end{equation*}
To achieve asymptotic normality, we use $(A.7)$ in \citet{cai2007trending},
which shows that $\sqrt{Th}H^{-1}r(t/T) \rightarrow^d N(0,V_{\gamma}(t/T))$.
Thus the desired result follows. \qed
\begin{lemma}
\thlabel{lemma1} Under assumptions A.1,A.2,A.3* and A.4, for $0 \leq j \leq
3 $, $h = O(T^{-1/5})$, and for $t/T \in [0,1]$, we have
\begin{equation}
h^{-j} S_j(t/T) - \mu_j M(t/T) = o_p(1),  \label{A2}
\end{equation}
and 
\begin{equation}
h^{-j} D_j(t/T) = o_p(h^2),  \label{A3}
\end{equation}
where $\mu_j = \int u^j k(u) du$.
\end{lemma}

\begin{proof}
See online appendix.
\end{proof}

\begin{lemma}
\thlabel{lemma2} Under assumptions A.1,A.2,A.3* and A.4, for $h =
O(T^{-1/5})$, and for $t/T \in [0,1]$, we have 
\begin{equation}
ThVar(H^{-1}r(t/T)) = V_{\gamma}(t/T) + o(1),  \label{A4}
\end{equation}
where $V_{\gamma}(t/T)=diag\{2\Omega(t/T)\nu_0,2\Omega(t/T)\nu_2\}$, $H =
diag\{I_{(d+1)},hI_{(d+1)}\}$, and $\nu_j = \int_{-1}^1 u^j k^2(u)du$.
\end{lemma}

\begin{proof}
See online appendix.
\end{proof}

\noindent \textbf{Proof of \thref{CVunif}}

\noindent For $\varepsilon > 0$, define 
\begin{equation*}
d_\varepsilon = \inf_{|\hat{h}_{CV} - h^{opt}| > T^{-1/5}\varepsilon} T^{-4/5}|CV(%
\hat{h}_{CV}) - CV(h^{opt})|
\end{equation*}
Then, 
\begin{align*}
&P(|\hat{h}_{CV} - h^{opt}| > T^{-1/5} \varepsilon) \leq P(T^{4/5} |CV(\hat{h}_{CV}) - CV(h^{opt})| > d_{\varepsilon}) \\
&= P(T^{4/5} |CV(\hat{h}_{CV}) - IMSCFE_L(h^{opt}) + IMSCFE_L(h^{opt}) -
CV(h^{opt})| > d_\varepsilon) \\
&\leq P(T^{4/5} |CV(\hat{h}_{CV}) - IMSCFE_L(\hat{h}_{CV})| > \frac{%
d_\varepsilon}{2}) + P(T^{4/5} |IMSCFE_L(h^{opt}) - CV(h^{opt})| > \frac{d_\varepsilon}{2%
}) \\
&\rightarrow 0,
\end{align*}
where $IMSCFE(h)_L$ is defined in \eqref{IMSCFE} and we have used the fact that $IMSCFE_L(h^{opt}) \\ \leq IMSCFE_L(\hat{h}%
_{CV})$ together with the fact that $CV(\hat{h}_{CV}) - CV(h^{opt}) \leq 0$
for the last inequality. The convergence follows from \thref{uniformCV}.

\begin{proposition}
\thlabel{uniformCV}
$CV(h) = \int_{0}^{1}\sigma (\tau )^{2}d\tau + IMSCFE(h)_L + o_p(T^{-4/5})$, uniformly in $h$.
\end{proposition}

\begin{proof}
Recall the definition of IMSCFE(h) from \eqref{IMSCFE}. We want to show that $|CV(h) - IMSCFE(h)| \\ = o_p(T^{-4/5})$ uniformly in $h$. To begin, note that
\begin{align}
\label{prop4main}
|CV(h) -  IMSCFE(h)| &\leq |CV(h) -  \hat{\tilde{CV}}(h)| + |\hat{\tilde{CV}}(h) - IMSCFE(h)| \notag \\
&\equiv d_1(h) + d_2(h),
\end{align}
where $\hat{\tilde{CV}}(h) = T^{-1} \sum_{s=1}^T \varepsilon_{s+1}^2 + T^{-1}
\sum_{s=1}^T \hat{D}_{s}^\top X_s X_s^\top \hat{D}_{s} - 2T^{-1}
\sum_{s=1}^T \varepsilon_{s+1} X_s^\top \hat{D}_{s}$, $\hat{D}_{s} = M^{-1}_s  T^{-1} \\ \sum_{\substack{ l=s-%
\floor*{Th}  \\ l \neq s}}^{s+\floor*{Th}} k_{ls}  X_l \varepsilon_{l,s}^*$, $\varepsilon_{l,s}^* = X_l^\top Q_{l,s}^* +
\varepsilon_{l+1}$, and $Q_{l,s}^* = \beta(l/T) - \beta(s/T) - (\frac{l-s}{T})
\beta^{^{\prime}}(s/T)$. We can complete the proof by showing that uniformly in $h$, $d_1(h) = o_p(T^{-4/5})$ and $d_2(h) = o_p(T^{-4/5})$.

We start with $d_1(h)$. Since our parameter space $\mathcal{H} \equiv [c_1 T^{-1/5}, c_2 T^{-1/5}]$
is compact, we can partition the space into intervals. In particular, let
there be $L_T = L(T)$ number of intervals $I_j \equiv I_{j,T}$ of length $%
l_{T} = (c_2 - c_1)T^{-1/5}/ L_T$, with midpoints denoted by $\tilde{h}_j
\equiv \tilde{h}_{j,T}$. We let the number of intervals diverge to infinity,
so set $L_T = O(T^{3/5})$ and by implication $l_T = O(T^{-4/5})$.

For a $\xi > 0$, 
\begin{align}  \label{unif}
P(\sup_{h \in \mathcal{H}} |d_1(h)| \geq T^{-4/5}\xi) &\leq
P(\max_{j=1,\ldots,L_T} \sup_{h\in\mathcal{H} \cap I_{j}} |d_1(h) - d_1(%
\tilde{h}_{j})| \geq T^{-4/5}\eta)  \notag \\
&+ P(\max_{j = 1,\ldots,L_T} |d_1(\tilde{h}_j)| \geq T^{-4/5} \eta)  \notag
\\
&= P_1^{d_1} + P_2^{d_1}
\end{align}
where $\eta = \xi/2$. Using \thref{orders}(i), we have $P_2^{d_1} \leq L_T \max_{j=1,\ldots,L_T} T^{8/5} \eta^{-2} E(d_1(\tilde{h}_j)^2) = o(1)$, since $E(d_1(\tilde{h}_j)^2)=O(T^{-12/5})$ and $L_T = O(T^{3/5})$. 

Next, we consider $P_1^{d_1}$. Let $h^* \equiv h^*_j$ be the solution to $%
sup_{h\in \mathcal{H} \cap I_j}|d_1(h) - d_1(\tilde{h}_j)|$. For
convenience, we assume that $L_T$ is chosen large enough so that $\floor*{T\tilde{h}_j} = \floor*{Th^*} \equiv \floor*{Th}$. Observe that 
\begin{align}  \label{ineq0}
P_1^{d_1} &\leq P(\max_{j=2,\ldots,L_T} |CV(h^*) - CV(\tilde{h}_j)| \geq
T^{-4/5} \eta/2) + P(\max_{j=2,\ldots,L_T} |\hat{\tilde{CV}}(\tilde{h}_j) - 
\hat{\tilde{CV}}(h^*)| \geq T^{-4/5} \eta/2)  \notag \\
&\leq 4 L_T T^{8/5} \eta^{-2} \max_{j=2,\ldots,L_T} \bigg\{ E([CV(h^*) - CV(%
\tilde{h}_j)]^2) + E([\hat{\tilde{CV}}(\tilde{h}_j) - \hat{\tilde{CV}}%
(h^*)]^2)\bigg\}.
\end{align}

We start with the second term of \eqref{ineq0}. Following the notation introduced in \thref{lemmaCV} and by Cauchy-Schwarz, we get
\begin{align*}
&E[(\hat{\tilde{CV}}(\tilde{h}_j)-\hat{\tilde{CV}}(h^*))^2] \\
&\leq 5\bigg\{ E[(CV_{bias}(\tilde{h}_j)-CV_{bias}(h^*))^2] + E[(CV_{variance}(\tilde{h}_j)-CV_{variance}(h^*))^2] \\ &+ E[(CV_{11}(\tilde{h}_j)-CV_{11}(h^*))^2] + E[(CV_{12}(\tilde{h}_j)-CV_{12}(h^*))^2] + E[(CV_{2}(\tilde{h}_j)-CV_{2}(h^*))^2] \bigg\} = O(T^{-14/5}),
\end{align*}
where $CV_{bias}(h) = T^{-3} \sum^3  Q_{m,s}^{*\top} X_m  X_m^\top  M_s^{-1} X_s  X_s^\top M_s^{-1} X_n X_n^\top Q_{n,s}^* k_{ms}  k_{ns}$, $CV_{11}(h) = T^{-3} \sum^3 \varepsilon_{m+1} X_m^\top \\ M_s^{-1}  X_s X_s^\top M_s^{-1} X_n X_n^\top
Q_{n,s}^* k_{ms} k_{ns}$, $CV_{12}(h) = T^{-3} \sum^3 Q_{m,s}^{*\top} X_m X_m^\top M_s^{-1} X_s  X_s^\top M_s^{-1} X_n
\varepsilon_{n+1} k_{ms} k_{ns}$, \\ $CV_{variance}(h) = T^{-3} \sum^3 \varepsilon_{m+1} X_m^\top M_s^{-1} X_s X_s^\top M_s^{-1}  X_n
\varepsilon_{n+1} k_{ms} k_{ns}$, $CV_2(h) = 2T^{-1}
\sum_{s=1}^T \varepsilon_{s+1} X_s^\top \hat{D}_{s}$, and $\sum^3 = \sum_{s=1}^T  \sum_{ m=s-\floor*{Th}, m \neq s}%
^{s+\floor*{Th}} \sum_{ n=s-\floor*{Th}, n \neq s}^{s+\floor*{%
Th}}$. The bound is obtained by applying \thref{orders}(ii)-(vi). 

For the first term in \eqref{ineq0}, we have
 \begin{align*}
&E[(CV(\tilde{h}_j)-CV(h^*))^2] = E[(CV(\tilde{h}_j) - \hat{\tilde{CV}}(\tilde{h}_j) + \hat{\tilde{CV}}(\tilde{h}_j) -\hat{\tilde{CV}}(h^*) + \hat{\tilde{CV}}(h^*)- CV(h^*))^2] \\
&\leq 3\bigg\{\underbrace{E[(CV(\tilde{h}_j) - \hat{\tilde{CV}}(\tilde{h}_j))^2]}_{\text{\thref{orders}(i)}} + \underbrace{ E[\hat{\tilde{CV}}(\tilde{h}_j) -\hat{\tilde{CV}}(h^*))^2]}_{\text{\thref{orders}(ii)-(vi)}} + \underbrace{ E[(\hat{\tilde{CV}}(h^*)- CV(h^*))^2]}_{\text{\thref{orders}(i)}} \bigg\} \\
&= O(T^{-12/5}).
\end{align*}
Hence, we conclude that $P_1^{d_1}=o(1)$, which implies that $d_1(h) = o_p(T^{-4/5})$.

Next, we focus on $d_2(h)$. Following the notation defined above, we have
\begin{align}
d_2(h) &\leq \bigg|T^{-1}\sum_{s=1}^T\varepsilon_{s+1}^2 - \int \sigma^2(\tau) d\tau\bigg| + \bigg|CV_{bias}(h) - \underbrace{\frac{h^4 \mu_2^2}{4} \int \Tr%
\{M(\tau)\beta^{^{\prime\prime}}(\tau)\beta^{^{\prime\prime}}(\tau)\}d\tau}_{\equiv \Psi_B(h)} \bigg| \notag \\
&+ \bigg| CV_{variance}(h) - \underbrace{\frac{2\nu_0}{Th} \int \Tr%
\{ M_\tau^{-1}V(\tau)\}d\tau}_{\equiv \Psi_V(h)} \bigg| + |CV_{11}(h)| + |CV_{12}(h)| + |CV_2(h)| \notag \\
&\equiv d_{\sigma} + d_{21}(h) + d_{22}(h) + |CV_{11}(h)| + |CV_{12}(h)| + |CV_2(h)|.
\end{align}

We start with $d_{21}(h)$. Next, define $P_2^{d_{21}}$ in an analogous manner to \eqref{unif}. Note that 
$E(d_{21}(h)^2) =
E(CV_{bias}^2(h)) - E(CV_{bias}(h))^2 + O(T^{-12/5}),
$
where from the proof of \thref{lemmaCV}, we have $E(CV_{bias}(h)) = \Psi_B(h) +  O(T^{-8/5})$. Let $Z_{smn} =
\beta^{^{\prime\prime}\top}_s X_m X_m^\top M_s^{-1} X_s X_s^\top  M_s^{-1}
X_n X_n^\top \beta^{^{\prime\prime}}_s 
k_{ms} k_{ns} (\frac{m-s}{T})^2 (\frac{n-s}{T})^2$ and observe that 
\begin{align*}
&Var_{bias} \equiv E(CV_{bias}^2) - E(CV_{bias})^2 \\
&= \frac{1}{16 T^6}\sum_{s=1}^T \sum_{\substack{ m=s-\floor*{Th}  \\ m \neq
s }}^{s+\floor*{Th}} \sum_{\substack{ n=s-\floor*{Th}  \\ n \neq s}}^{s+%
\floor*{Th}} \sum_{a=1}^T \sum_{\substack{ b=a-\floor*{Th}  \\ b \neq a}}^{a+%
\floor*{Th}} \sum_{\substack{ c=a-\floor*{Th}  \\ c \neq a}}^{a+\floor*{Th}} %
\bigg[E(Z_{smn} Z_{abc}) - E(Z_{smn})E(Z_{abc})\bigg].
\end{align*}
Consider the case where $m < s < n < b < a < c$ and $n-s = \bar{D}$ is the
largest of the adjacent differences. Label this case $(*)$. Then we have, 
\begin{align*}
&Var_{bias}^{*} \\
&\leq \frac{1}{16 T^6} {\sum \ldots \sum}_{\substack{ m<s<n<b<a<c  \\ n - s
= \bar{D}}} \Theta_*^{1/1+\delta} \beta^*(n-s)^{\delta/1+\delta} k_{ms}
k_{ns} \bigg(\frac{m-s}{T}\bigg)^2 \bigg(\frac{n-s}{T}\bigg)^2 \\
&\times k_{ba} k _{ca} \bigg(\frac{b-a}{T}\bigg)^2 \bigg(\frac{c-a}{T}\bigg)^2 \\
&\leq \frac{C h^6}{16T^4} \sum_s \sum_a \underbrace{\sum_{n<s} (n-s)
\beta^*(n-s)^{\delta/1+\delta}}_{<\infty} \underbrace{k\bigg(\frac{n-s}{Th}%
\bigg)}_{\leq k(0)} \underbrace{\bigg(\frac{n-s}{Th}\bigg)^2}_{\leq 1} = O\bigg(\frac{h^6}{T^2}\bigg),
\end{align*}
where $\Theta_* = \max(\Theta_*^a, \Theta_*^b)$, $\Theta_*^a = \sup_{s,
\ldots,c} \int Z_{smn}Z_{abc} dF(X_m, \ldots,X_c)$, and $\Theta_*^b =
\sup_{s,\ldots,c} \int \int Z_{smn} 
Z_{abc} \\ dF(X_m,X_s,X_n) dF(X_a,X_b,X_c)$. The largest order can be obtained by taking the maximum adjacent difference
to be $s-m$ (i.e. $\bar{D} = s - m$). In this case, the order would be $%
O(h^7/T) = O(T^{-12/5})$. Hence we
conclude that $E(d_{21}(h)^2) = Var_{bias} = O(h^7/T)=O(T^{-12/5})$, and so again $
P_2^{d_{21}} \leq L_T \eta^{-2} T^{8/5} O(T^{-12/5}) = o(1)$.

Similarly, we can show that $%
E(d_{22}(h)^2) = O((Th)^{-3})$, $E(CV_{11}(h)^2) = O(h^2/T^2)$, and $%
E(CV_{12}(h)^2) = O(h^2/T^2)$ which are all equivalent to the order of $%
T^{-12/5}$. For $CV_2$, we have $E(CV_2(h)^2) = O(h^3/T^2) = O(T^{-13/5})$. Hence, $P_2^i = o(1)$ where $i$ is $d_{21}$, $CV_{11}$, $CV_{12}$, $d_{22}$, or $CV_2$.

For $d_{21}$ and $d_{22}$, notice that both $|\Psi_B(h^*) - \Psi_B(%
\tilde{h}_j)|$ and $|\Psi_V(h^*) - \Psi_V(\tilde{h}_j)|$ are $O(T^{-7/5})$.
Since, $\Psi_B(h^*) - \Psi_B(\tilde{h}_j) = C(h^{*4} - \tilde{h}_j^4) =
C(h^{*2} + \tilde{h}_j^2)(h^{*} + \tilde{h}_j)(h^{*} - \tilde{h}_j) = O(h^3 \times T^{-4/5}) = O(T^{-7/5})$, where $C$ is some constant. Focusing on $d_{21}(h)$, we have
\begin{align*}
&P_1^{d_{21}} = P(\max_{j=1,\ldots,L_T} \sup_{h \in \mathcal{H} \cap I_j}
|d_{21}(h) - d_{21}(\tilde{h}_j)| \geq T^{-4/5} \eta) \\
&\leq P(\max_{j=1,\ldots,L_T} |CV_{bias}(h^*) - CV_{bias}(\tilde{h}_j)| \geq
T^{-4/5} \eta/2) + P(\max_{j=1,\ldots,L_T} |\Psi_B(\tilde{h}_j) -
\Psi_B(h^*)| \geq T^{-4/5} \eta/2) \\
&\leq 4 L_T T^{-8/5} \eta^{-2} O(T^{-14/5}) = o(1),
\end{align*}
where we have used \thref{orders}(ii).
By repeating these steps with relevant parts of \thref{orders} for the other quantities, we can conclude that $P_1^i = o(1)$
where $i$ is $d_{21}$, $CV_{11}$, $CV_{12}$, $d_{22}$, or $CV_2$.

Lastly, note that $E(d_\sigma^2)=O(T^{-2})$ so that $P(|d_\sigma| \geq T^{-4/5}\xi)=o(1)$ for any $\xi>0$, which together with the results above, imply that $ P(\sup_{h\in\mathcal{H}}|d_2(h)|\geq T^{-4/5}\xi) = o(1),$ thereby completing the proof.
\end{proof}

\begin{lemma}
\thlabel{CVapprox}
$CV(h) = \hat{\tilde{CV}}(h) + o_p(T^{-6/5})$, where $\hat{\tilde{CV}}(h) = T^{-1} \sum_{s=1}^T \varepsilon_{s+1}^2 + T^{-1}
\sum_{s=1}^T \hat{D}_{s}^\top X_s X_s^\top \hat{D}_{s} - 2T^{-1}
\sum_{s=1}^T \varepsilon_{s+1} X_s^\top \hat{D}_{s} $, $\hat{D}_{s} = M^{-1}_s \ T^{-1} \sum_{\substack{ l=s-%
\floor*{Th}  \\ l \neq s}}^{s+\floor*{Th}} k_{ls}X_l \varepsilon_{l,s}^*$, $\varepsilon_{l,s}^* = X_l^\top Q_{l,s}^* +
\varepsilon_{l+1}$, and $Q_{l,s}^* = \beta(l/T) - \beta(s/T) - (\frac{l-s}{T})
\beta^{^{\prime}}(s/T)$.
\end{lemma}

\begin{proof}
See online appendix.
\end{proof}

\begin{lemma}
\thlabel{lemmaCV}
$E(\hat{\tilde{CV}}(h))= \int_{0}^{1}\sigma^2 (\tau )d\tau + IMSCFE(h)_L + o(T^{-4/5}),$ where $IMSCFE(h)_L$ is defined in \eqref{IMSCFE}.
\end{lemma}

\begin{proof}
See online appendix.
\end{proof}

\begin{lemma}
\thlabel{orders}
The following are bounded as (i) $E[(CV(h) - \hat{\tilde{CV}}(h))^2] = O(T^{-12/5})$, (ii) $E[(CV_{bias}(\tilde{h}_j)-CV_{bias}(h^*))^2] = O(T^{-14/5})$, (iii) $E[(CV_{variance}(\tilde{h}_j)-CV_{variance}(h^*))^2] = O(T^{-14/5})$, (iv) $E[(CV_{11}(\tilde{h}_j)-CV_{11}(h^*))^2] = O(T^{-14/5})$, (v) $E[(CV_{12}(\tilde{h}_j)-CV_{12}(h^*))^2] = O(T^{-14/5})$, and (vi) $E[(CV_{2}(\tilde{h}_j)-CV_{2}(h^*))^2]\\= O(T^{-3})$, whose definitions are provided in the proof of \thref{uniformCV}.
\end{lemma}

\begin{proof}
See online appendix.
\end{proof}

Next, we shall prove \thref{fsprop},  \thref{ssoracle}, and \thref{oracleprop} from Section \protect\ref{high.asymp}. To recall, we have $d \ll T$ relevant forecasts out of $p_T+1$ candidates. Due to the local linear structure of our estimation, we also impose sparsity on the gradient regressors. Specifically, we assume that $d_1 \leq d$ gradient coefficients are non-zero, while the rest are. Before proving the results proper, we state a lemma which is used in the following proofs and is the analogue of lemma C.1 in \cite{li2015model}.

\begin{lemma}
\thlabel{lemmac1} Define 
\begin{align*}
Z_{sj}(\tau, l) &\equiv \frac{1}{h} k\bigg(\frac{s/T-\tau}{h}\bigg) (y_{s+1}
- \alpha_0(s/T)^\top X_s) \bigg[-\bigg(\frac{s/T-\tau}{h}\bigg)^l X_{sj} %
\bigg]
\end{align*}
for $l=0,1$ and $X_{sj}$ is the $j^{th}$ element of $X_s$. Suppose the
assumptions of \thref{fsprop} are true, then we have 
\begin{equation*}
\max_{1 \leq j \leq p_T+1} \sup_{\tau \in [0,1]} \bigg|T^{-1} \sumTau Z_{sj}(\tau,l) %
\bigg| = O_p \bigg(\sqrt{\frac{\log h^{-1}}{Th}}\bigg),
\end{equation*}
for $l = 0,1$.
\end{lemma}

\begin{proof}
See online appendix.
\end{proof}

\noindent \textbf{Proof of \thref{fsprop}}

\noindent To prove our result, we want to show that 
\begin{equation}
\max_{t} \| \tilde{\alpha}_{0t} - \alpha_{0t}\| + \max_{t} h\| \tilde{\alpha}_{1t} - \alpha_{1t} \| = O_p(\sqrt{d}\lambda_1),
\label{fsresult}
\end{equation}
where $d$ is the number of relevant forecasts. Let $\tilde{\alpha}_{it}$ denote the initial estimates from the first-stage
in \eqref{prelim} at time $t$ for $i=0,1$, and $\tilde{\alpha}_{it}^j$ be
the $j^{th}$ element or coefficient. Define the following error terms: $%
d_{0t}^j=\tilde{\alpha}_{0t}^j-\alpha_{0t}^j$ and $d_{1t}^j=h\{\tilde{\alpha}%
_{1t}^j-\alpha_{1t}^j\}$. We first show that uniformly in $t=1,\ldots,T$, we
have that 
\begin{equation}
\max\bigg\{\sum_{j=d+1}^{p_T+1} |d_{0t}^j| ,
\sum_{j=d_1+1}^{p_T+1}|d_{1t}^j|\bigg\} \leq (1+\delta) \bigg\{ \sum_{j=1}^d
|d_{0t}^j|+ \sum_{j=1}^{d_1} |d_{1t}^j| \bigg\}  \label{strictconcaverest}
\end{equation}
for some $\delta > 0$. \cite{bickel2009simultaneous} show that such a cone constraint for the Lasso error is essential for consistency.

Let $\mathcal{L}_t^{fs}(\alpha_{0t},\alpha_{1t})$ represent the first term
in \eqref{prelim}, and let the entire expression be $Q^{fs}(\alpha_{0t},%
\alpha_{1t})$. By construction, we have $Q^{fs}_t(\tilde{\alpha}_{0t},\tilde{%
\alpha}_{1t}) \leq Q^{fs}_t(\alpha_{0t},\alpha_{1t})$ and thus, 
\begin{equation}
\mathcal{L}_t^{fs}(\tilde{\alpha}_{0t},\tilde{\alpha}_{1t})-\mathcal{L}%
_t^{fs}(\alpha_{0t},\alpha_{1t}) \leq \lambda_1 \bigg\{ \sum_{j=1}^{p_T+1}
|\alpha_{0t}^j| - \sum_{j=1}^{p_T+1} |\tilde{\alpha}_{0t}^j| \bigg\} +
\lambda_2 \bigg\{ \sum_{j=1}^{p_T+1} |h\alpha_{1t}^j| - \sum_{j=1}^{p_T+1} |%
\tilde{h\alpha}_{1t}^j| \bigg\} .  \label{b0}
\end{equation}
Since $\mathcal{L}_t^{fs}$ is a convex function, we have that 
\begin{equation}
D_t^\top \dot{\mathcal{L}}_t^{fs}(\alpha_{0t}, \alpha_{1t}) \leq \mathcal{L}%
_t^{fs}(\tilde{\alpha}_{0t},\tilde{\alpha}_{1t})-\mathcal{L}%
_t^{fs}(\alpha_{0t},\alpha_{1t})  \label{b1}
\end{equation}
where $D_t =
(d_{0t}^1,\ldots,d_{0t}^{p_T+1},d_{1t}^1,\ldots,d_{1t}^{p_T+1})^\top$ and 
\begin{equation*}
\underbrace{\dot{\mathcal{L}}_t^{fs}(\alpha_{0t}, \alpha_{1t})}_{2(p_T+1)
\times 1} = 2T^{-1} \sum_{\substack{ s=t-\floor*{Th}  \\ s \neq t}}^{t+\floor%
*{Th}} k_{st}\bigg[y_{s+1} - \alpha_{0}(t/T)^\top X_s - \alpha_{1}(t/T)^\top %
\bigg(\frac{s-t}{T}\bigg)X_s\bigg] \times%
\begin{bmatrix}
-X_s \\ 
-\bigg(\frac{s-t}{Th}\bigg) X_s%
\end{bmatrix}%
.
\end{equation*}
From \thref{lemmac1}, we get 
\begin{equation*}
\max_{1\leq j \leq p_T+1} \sup_{t/T} \bigg|T^{-1} \sum_{\substack{ s=t-\floor%
*{Th}  \\ s \neq t}}^{t+\floor*{Th}} k_{st} \bigg[y_{s+1} -
\alpha_{0}(s/T)^\top X_s\bigg] \bigg\{-\bigg(\frac{s-t}{Th}\bigg)^l X_{sj} %
\bigg\} \bigg| = O_p\bigg(\sqrt{\frac{\log{h^{-1}}}{Th}}\bigg)
\end{equation*}
for $l=0,1$. Denote $\dot{\mathcal{L}}_{t,j}^{fs}(\alpha_{0t}, \alpha_{1t})$
to be the $j^{th}$ component in $\dot{\mathcal{L}}_{t}^{fs}(\alpha_{0t},
\alpha_{1t})$. Then note that by the triangle inequality and the fact that $%
x \geq -|x|$, we can write 
\begin{equation}
D_t^\top \dot{\mathcal{L}}_t^{fs}(\alpha_{0t}, \alpha_{1t}) \geq -
\sum_{j=1}^{p_T+1} (|d_{0t}^j|+|d_{1t}^j|) \max_{j} |\dot{\mathcal{L}}%
_{t,j}^{fs}(\alpha_{0t}, \alpha_{1t})|.  \label{b2}
\end{equation}
Furthermore, 
\begin{equation}
\max_j |\dot{\mathcal{L}}_{t,j}^{fs}(\alpha_{0t}, \alpha_{1t})| = O_p\bigg(%
\sqrt{\frac{\log{h^{-1}}}{Th}} + d h^2\bigg) = o_p(\lambda_1),  \label{b3}
\end{equation}
where the $h^2$ comes from the bias in the derivative term, and the last
equality is from assumption HD.3(i).

Next, we focus on the $d_{0t}$ term $\sum_{j=1}^{p_T+1} |d_{0t}^j| = \sum_{j=1}^{d} |d_{0t}^j|+
\sum_{j=d+1}^{p_T+1} |d_{0t}^j|  \label{b4}
$,
and note that 
\begin{align}
\lambda_1 \bigg\{ \sum_{j=1}^{p_T+1} |\alpha_{0t}^j| - \sum_{j=1}^{p_T+1} |%
\tilde{\alpha}_{0t}^j|\bigg\} &= \lambda_1 \bigg\{ \sum_{j=1}^{d}
|\alpha_{0t}^j| - \sum_{j=1}^{d} |\tilde{\alpha}_{0t}^j| -
\sum_{j=d+1}^{p_T+1} |\tilde{\alpha}_{0t}^j| \bigg\}  \notag \\
&\leq \lambda_1 \bigg\{ \sum_{j=1}^d |d_{0t}^j| - \sum_{j=d+1}^{p_T+1}
|d_{0t}^j| \bigg\}.  \label{b5}
\end{align}
By combining \eqref{b0} to \eqref{b5}, we get 
\begin{equation*}
(1-o_p(1))\sum_{j=d+1}^{p_T+1} |d_{0t}^j| \leq (1 + o_p(1))\sum_{j=1}^d
|d_{0t}^j| \leq (1+\delta) \sum_{j=1}^d |d_{0t}^j|, \text{ for some } \delta > 0.
\end{equation*}
We can repeat the same procedure for $d_{1t}$, and subsequently arrive at %
\eqref{strictconcaverest}.

Next, we want to determine the existence of a local minimizer within the neighborhood of $(\alpha_{0t}, \alpha_{1t})$ such that this cone constraint is satisfied. Define $v_1=(v_{11},\ldots,v_{1p_T+1})^\top$, and $v_2=(v_{21},%
\ldots,v_{2p_T+1})^\top$, where $v_1$ and $v_2 \in \mathbb{R}^{p_T+1}$.
Define $S(C_0)=\{v=(v_1^\top,v_2^\top)^\top: \|v_1\|^2+\|v_2\|^2=C_0,
\sum_{m=1}^{p_T+1} (|v_{1m}|+|v_{2m}|) \leq
2(1+\delta)\sum_{m=1}^{d}(|v_{1m}|+|v_{2m}|) \}$ for some $\delta > 0$,
and $C_0 > 0$ is an arbitrary constant.

Let $\kappa_T=\sqrt{d}\lambda_1$. Recall that $d$ is allowed to increase with $T$, hence the dependence of $\kappa$ on $T$. Now, we show that there
exists a minimizer in the interior of $\{(\alpha_{0t}+\kappa_T v_1,
h\alpha_{1t}+\kappa_T v_2): v = (v_1^\top, v_2^\top)^\top \in S(C_0)\}$. To
do this, consider 
\begin{align}
&Q_t^{fs}(\alpha_{0t}+\kappa_T v_1, \alpha_{1t}+\kappa_T v_2/h) -
Q_t^{fs}(\alpha_{0t},\alpha_{1t})  \label{b6} \\
&= \mathcal{L}_t^{fs}(\alpha_{0t}+\kappa_T v_1, \alpha_{1t}+\kappa_T v_2/h)
- \mathcal{L}_t^{fs}(\alpha_{0t},\alpha_{1t})  \notag \\
&+ \lambda_1 (|\alpha_{0t}+\kappa_T v_1| - |\alpha_{0t}|) + \lambda_2
(|h\alpha_{1t}+\kappa_T v_2| - |h\alpha_{1t}|)  \notag \\
&\equiv I_1 + I_2 + I_3  \notag
\end{align}
The proof for $I_2$ and $I_3$ follows that of \cite{li2015model}, and we can
show that $I_2 = O_p(\kappa_T^2)\|v_1\| + \lambda_1
\sum_{j=d+1}^{p_T+1}|\kappa_T v_{1j}|$ and $I_3 = O_p(\kappa_T^2)\|v_2\| +
\lambda_2 \sum_{j=d_1+1}^{p_T+1} |\kappa_T v_{2j}|$. For $I_1$, we rewrite $\mathcal{L}_t^{fs}(\alpha_{0t},\alpha_{1t})$ in its matrix form: $
T^{-1}(Y_t - Q_t \gamma_t)^\top K_t (Y_t-Q_t\gamma_t)
$, 
where $Y_t$, $Q_t$ and $K_T$ are defined in Section \ref{notation}, and $%
\gamma_t = (\alpha_{0t}^\top,\alpha_{1t}^\top)^\top$ as before. Also, set $%
\tilde{\gamma}_t = ((\alpha_{0t}+\kappa_T v_1)^\top, (\alpha_{1t}+\kappa_T
v_2/h)^\top)^\top$, and $V=(\kappa_T v_1 ^\top, \kappa_T v_2/h ^\top)^\top
=\kappa_T H \cdot (v_1^\top, v_2^\top)^\top = \kappa_T H v$, with $H$
defined in assumption HD.2 (also recall that $Q^h_t = Q_t H$). Then, 
\begin{align}
I_1 &= T^{-1}[(Y_t - Q_t \gamma_t)^\top K_t (Y_t-Q_t\gamma_t) - (Y_t - Q_t 
\tilde{\gamma}_t)^\top K_t (Y_t-Q_t\tilde{\gamma}_t)]  \notag \\
&= T^{-1}[-2Z_t^\top K_t Q_t V + V^\top Q_t^\top K_t Q_t V]  \notag \\
&= T^{-1}[-2\kappa_T Z_t^\top K_t Q_t H v + \kappa_T^2 v^\top Q_t^{h\top}K_t
Q_t^h v]  \label{b7}
\end{align}
where we have arrived at the second last line by defining $Y_t - Q_t
\gamma_t \equiv Z_t$. Note that $-T^{-1}Z_t^\top K_t Q_t H$ is a $1 \times
2(p_T+1)$ vector and can be written as 
\begin{align*}
-T^{-1}Z_t^\top K_t Q_t H = \bigg[-T^{-1} \sum_{\substack{ s=t-\floor*{Th} 
\\ s \neq t}}^{t+\floor*{Th}} &k_{st}(y_{s+1}-q_{st}^\top\gamma_t)X_{s1},
\ldots, \\ &-T^{-1} \sum_{\substack{ s=t-\floor*{Th}  \\ s \neq t}}^{t+\floor*{Th%
}} k_{st}\bigg(\frac{s-t}{Th}\bigg)(y_{s+1}-q_{st}^\top
\gamma_t)X_{s(p_T+1)} \bigg]
\end{align*}
where $q_{st}$ is analogous to the notation defined in Section \ref{nonpar}, and $X_{sj}$ is the $j^{th}$ variable in $X_s$. Hence, we can show that the first term in %
\eqref{b7} is $o_p(\kappa_T^2)\|v\|$ uniformly in $t=1,\ldots,T$ using the
Cauchy-Schwarz inequality and \thref{lemmac1}. By assumption HD.2, we have
that $v^\top Q_t^{h\top} K_t Q_t^h v \geq \rho_1$, and thus the second term
is strictly positive. So we conclude that the leading term in \eqref{b6} is
positive in probability, and a local minimizer exists in the cone
constraint. By the convexity of the objective function, the local minimizer
is thus the global minimizer.

\qed 

\noindent \textbf{Proof of \thref{ssoracle}}

\noindent To facilitate the proof, we introduce the following penalized oracle problem. We define the \textit{biased} oracle estimator as the solution to
\begin{align}
\underset{\Upsilon \in \mathbb{R}^{T\times d_1}}{\arg
\min } \quad &T^{-1} \bigg(\overline{Y}-\sum_{j=1}^{d}\Xi _{j}\alpha
_{0,j}-\sum_{j=1}^{d_1}\Xi_{j+p_T+1}\alpha _{1,j}\bigg)^{\top }%
\overline{K}\bigg(\overline{Y}-\sum_{j=1}^{d}\Xi _{j}\alpha
_{0,j}-\sum_{j=1}^{d_1}\Xi_{j+p_T+1}\alpha _{1,j}\bigg) \notag \\
&+ \sum_{j=1}^{d}p_{\lambda _{3}}^{' }\left( \left\Vert \tilde{B}_{j}\right\Vert \right) \left\Vert \alpha _{0,j}\right\Vert
+\sum_{j=1}^{d_1}p_{\lambda _{4}}^{' }\left( \tilde{D}_{j}\right)
\left\Vert h\alpha _{1,j}\right\Vert, \label{biasedoracle}
\end{align}
which is similar to \eqref{llscad} but with irrelevant forecasts removed prior to optimization. Define the $T \times 2(p_T+1)$ matrix $\hat{\Upsilon}^b = (\hat{\Upsilon}_0, \mathbf{0}_{T \times ((p_T+1)-d)}, \hat{\Upsilon}_1, \mathbf{0}_{T \times ((p_T+1)-d_1)})$, where $\mathbf{0}_x$ are zero matrices of dimension $x$, and $(\hat{\Upsilon}_0, \hat{\Upsilon}_1)$ is obtained from minimizing the biased oracle problem \eqref{biasedoracle}. Note that $\hat{\Upsilon}_0$ and $\hat{\Upsilon}_1$ are $T \times d$ and $T\times d_1$ matrices that correspond to relevant regressors in $\alpha_0$ and $\alpha_1$ respectively.

The proof of \thref{ssoracle} consists of the proofs of the following two propositions. Specifically, \thref{selectconst} shows the the two-stage estimate and the biased oracle solution are equivalent with probability approaching 1, which implies \eqref{select}. Furthermore, \thref{biasedconst} shows that the biased oracle estimate is consistent for the true parameters. The desired result in \eqref{est} follows.

\qed 
\begin{proposition}
\thlabel{biasedconst}
Define the vector $A_i = (\alpha_{i,1}^\top, \ldots, \alpha_{i,p_T+1}^\top)^\top$ for $i=0,1$. Let $(\hat{\Upsilon}%
_0, \hat{\Upsilon}_1)$ solve the oracle problem in \eqref{biasedoracle}, and let $\hat{\Upsilon}_0^o =  vec(\hat{\Upsilon}_0, \mathbf{0}_{T \times ((p_T+1)-d)})$, and $\hat{\Upsilon}_1^o = vec(\hat{%
\Upsilon}_1, \mathbf{0}_{T \times ((p_T+1)-d_1)}) $. Under the assumptions of \thref{ssoracle}, we have
\begin{equation}  \label{convergerate}
T^{-1/2} \| \hat{\Upsilon}_0^o - A_0 \| = O_p(\sqrt{d/Th}), \quad T^{-1/2}
\| \hat{\Upsilon}_1^o - A_1\| = O_p(\sqrt{d/Th^3}).
\end{equation}
\end{proposition}

\begin{proof}

Define the following $T(p_T+1) \times 1$ vectors: 
\begin{equation*}
U_1 = (u_{11}^\top, \ldots, u_{1d}^\top, \mathbf{0}^\top, \ldots, \mathbf{0}%
^\top)^\top, \quad U_2 = (u_{21}^\top, \ldots, u_{2d_1}^\top, \mathbf{0}%
^\top, \ldots, \mathbf{0}^\top)^\top,
\end{equation*}
where $u_{ij} \in \mathbb{R}^{T}$ for $i=1,2$.

Define the set $S^\diamond(C_1) = \{U=(U_1^\top, U_2^\top)^\top:
\|U_1\|+\|U_2\| = TC_1 \}$, where $C_1$ is an arbitrary positive constant,
and $\kappa_T^\diamond = \sqrt{d/Th}$. For convenience, let the objective function
in \eqref{llscad} be $Q^{ss}(A_0,A_1)$, where $\alpha_{i,j}$ is defined in \eqref{llscad}. Similar to the proof above, for a $U \in S^\diamond(C_1)$, we can write 
\begin{align}
&Q^{ss}(A_0 + \kappa_T^\diamond U_1,A_1 + \kappa_T^\diamond U_2/h) -
Q^{ss}(A_0,A_1)  \label{b8} \\
&= \mathcal{L}^\diamond(A_0 + \kappa_T^\diamond U_1,A_1 + \kappa_T^\diamond
U_2/h) - \mathcal{L}^\diamond(A_0,A_1) \   \notag \\
&+ \sum_{j=1}^{p_T+1} p^{^{\prime}}_{\lambda_3} (\| \tilde{B}_j \|)
\|\alpha_{0,j} + \kappa_T^\diamond u_1^j \| - \sum_{j=1}^{p_T+1}
p^{^{\prime}}_{\lambda_3} (\| \tilde{B}_j \|) \|\alpha_{0,j} \|  \notag \\
&+ \sum_{j=1}^{p_T+1} p^{^{\prime}}_{\lambda_3} (\| \tilde{D}_j \|)
\|h\alpha_{1,j} + \kappa_T^\diamond u_2^j\| - \sum_{j=1}^{p_T+1}
p^{^{\prime}}_{\lambda_3} (\| \tilde{D}_j \|) \|h\alpha_{1,j} \|  \notag \\
&\equiv I_1^\diamond + I_2^\diamond + I_3^\diamond  \notag
\end{align}
where $u_1^j$ is the $j^{th}$ row of the matrix obtained by applying an
inverse vectorization of $U_0$. In other words, let $vec^{-1}(U_1)$ be the $%
T \times (p_T+1)$ matrix obtained by sequentially unstacking $U_1$, and let $%
[vec^{-1}(U_1)]_j$ denote its $j^{th}$ row, then $u_1^j = [vec^{-1}(U_1)]_j$%
. This is analogous for $u_2^j$ with $U_2$. The same can be said for the
notation of $\mathcal{L}^\diamond(A_0,A_1)$, which previously took as
arguments, the matrices $\alpha_0$ and $\alpha_1$ in \eqref{llmatrix}. Here, 
$\alpha_0 = vec^{-1}(A_0)$ and $\alpha_1 = vec^{-1}(A_1)$.

The proof of $I_2^\diamond$ and $I_3^\diamond$ follows that of \cite%
{li2015model}, which requires in particular, the assumption HD.3(iii).
Thus, they can be shown to be $o_p(\kappa_T^{\diamond 2})\|U_1 \|^2$ and $%
o_p(\kappa_T^{\diamond 2})\|U_2 \|^2$ respectively.

Let $\Xi_0 = (\Xi_1, \ldots, \Xi_{p_T+1})$ and $\Xi_1 = (\Xi_{p_T+2},
\ldots, \Xi_{2(p_T+1)})$. For $I_1^\diamond$, we adopt the same approach as
the proof above, 
\begin{align}
I_1^\diamond &= T^{-1}\{[\overline{Y} - \Xi_0 (A_0 + \kappa_T^\diamond U_1)
- \Xi_1 (A_1 + \kappa_T^\diamond U_2/h)]^\top \overline{K} [\overline{Y} -
\Xi_0 (A_0 + \kappa_T^\diamond U_1) - \Xi_1 (A_1 + \kappa_T^\diamond U_2/h)]
\notag \\
&- [\overline{Y} - \Xi_0 (A_0) - \Xi_1 (A_1)]^\top \overline{K} [\overline{Y}
- \Xi_0 (A_0) - \Xi_1 (A_1)]\}  \notag \\
&= T^{-1}\{[2 \overline{Y} - 2 \Xi A - \Xi U^*]^\top \overline{K} [- \Xi
U^*]\}  \label{b9}
\end{align}
where $\Xi=(\Xi_1, \ldots, \Xi_{2(p_T+1)})$ is $2\floor*{Th}T \times
2(p_T+1)T$, $A =(A_0^\top, A_1^\top)^\top$ and $U^* =
\kappa_T^\diamond(U_1^\top, U_2^\top/h)^\top = \kappa_T^\diamond H^*
(U_1^\top, U_2^\top)^\top = \kappa_T^\diamond H^* U$. $H^* =
diag\{I_{(p_T+1)T \times (p_T+1)T}, \mathbf{h}^{-1*}\}$ and $\mathbf{h}%
^{-1*} $ is a $(p_T+1)T \times (p_T+1)T$ diagonal matrix with diagonal
elements fixed at $1/h$. We can show that $\Xi H^* = (\Xi_1, \ldots,
\Xi_{p_T+1}, \newline
h^{-1} \Xi_{p_T+2}, \ldots, h^{-1}\Xi_{2(p_T+1)})$. Note that for $p_T+2
\leq j \leq 2(p_T+1)$, we have $e_{j,2(p_{T}+1)} h^{-1} = H e_{j,2(p_{T}+1)}$,
where $H = diag\{I_{p_T+1\times p_T+1}, \mathbf{h}^{-1}\}$ is defined in
assumption HD.2. This implies that $h^{-1} \Xi_j = \Xi_j^h$, and $\Xi
H^* = \Xi^h$, where $\Xi^h = (\Xi_1, \ldots, \Xi_{p_T+1}, \Xi_{p_T+2}^h, \ldots \Xi_{2(p_T+1)}^h)$ is a $2\floor*{Th}T \times 2T(p_T+1)$ matrix, $\Xi_i^h$ is defined as in \eqref{Xi} but with $Q_i$ replaced by $Q_i^h$ as defined in HD.2, and $\overline{K}$ is defined in the discussion on \eqref{llmatrix}.

Using the same strategy as above, from \eqref{b9}, we get 
\begin{align}
I_1^\diamond &= T^{-1} \{ -2\kappa_T^\diamond\overline{Y}^\top \overline{K}
\Xi^h U + 2\kappa_T^\diamond A^\top \Xi^\top \overline{K} \Xi^h U +
(\kappa_T^\diamond)^2 U^\top \Xi^{h\top} \overline{K} \Xi^h U \}  \notag \\
&= -2\kappa_T^\diamond T^{-1}\overline{Z}^\top \overline{K} \Xi^h U +
(\kappa_T^\diamond)^2 T^{-1} U^\top \Xi^{h\top} \overline{K} \Xi^h U,
\label{b10}
\end{align}
where $\overline{Z} = \overline{Y} - \Xi A$. From HD.2 we know that the
second term in \eqref{b10} is strictly positive. The first term is a $1
\times 2(p_T+1)T$ vector of the following form: 
\begin{align}  \label{b11}
&T^{-1}\overline{Z}^\top \overline{K} \Xi^h = T^{-1} \bigg[ \sum_{s=1-\floor*%
{Th},s \neq 1}^{1+\floor*{Th}} k_{s1} (y_{s+1}-q_{s1}^\top \gamma_1) X_{s1},
\ldots, \sum_{s=T-\floor*{Th}, s \neq T}^{T+\floor*{Th}} k_{sT} (y_{s+1} -
q_{sT}^\top \gamma_T) X_{s1}, \ldots, \notag \\
&\sum_{s=1-\floor*{Th}, s \neq 1}^{1+\floor*{Th}} k_{s1}
(y_{s+1}-q_{s1}^\top \gamma_1) \bigg(\frac{s-1}{Th}\bigg)X_{s(p_T+1)},
\ldots, \sum_{s=T-\floor*{Th}, s \neq T}^{T+\floor*{Th}} k_{sT}
(y_{s+1}-q_{sT}^\top \gamma_1) \bigg(\frac{s-T}{Th}\bigg)X_{s(p_T+1)} \bigg].
\end{align}
Again, we can apply \thref{lemmac1} to this term, in a similar fashion to %
\eqref{b7}, which implies that it is $O_p(\kappa_T^\diamond T^{1/2})\|U\|$.
Hence, we conclude that $I_1^\diamond$ (and thus \eqref{b8}) is positive in probability.

This suggests that for any $\varepsilon > 0$ and for a sufficiently large
choice of $C_1$ we have, 
\begin{equation*}
P\bigg(\inf_{(U_1,U_2) \in S^\diamond(C_1)} Q^{ss}(A_0 + \kappa_T U_1, A_1 +
\kappa_T U_2/h) > Q^{ss}(A_0, A_1) \bigg) \geq 1 - \varepsilon,
\end{equation*}
which implies \eqref{convergerate} .
\end{proof}
\begin{proposition}
\thlabel{selectconst}
Under the assumptions of \thref{ssoracle}, $\hat{\Upsilon}^b$ is the solution of \eqref{llscad} with probability approaching one. 
\end{proposition}

\begin{proof}

For convenience, let parameter with hats indicate estimates from $\hat{%
\Upsilon}^b$. Define $\hat{Z} = \overline{Y}-\sum_{j=1}^{p_{T}+1}\Xi _{j}%
\hat{\alpha}_{0,j} \\ -\sum_{j=1}^{p_{T}+1}\Xi _{j+p_{T}+1}^h h\hat{\alpha}%
_{1,j} $ and similarly, $\hat{Z}^o = \overline{Y}-\sum_{j=1}^{d}\Xi _{j}\hat{%
\alpha}_{0,j}-\sum_{j=1}^{d_1}\Xi _{j+p_{T}+1}^h h\hat{\alpha}_{1,j}$.
Then, to show that $\hat{\Upsilon}^b$ solves \eqref{llscad}, we need to
check that it satisfies the following Karush-Kuhn-Tucker conditions: (i) $T^{-1}\Xi_j^\top \overline{K} \hat{Z} + p_{\lambda_3}^{^{\prime}} (\| 
\tilde{B}_j \|) \frac{\hat{\alpha}_{0,j}}{\|\hat{\alpha}_{0,j}\|} = \mathbf{0%
}_{T\times 1} $, (ii) $T^{-1}\Xi_{j+p_T+1}^{h\top} \overline{K} \hat{Z} +
p_{\lambda_4}^{^{\prime}} (\tilde{D}_j) \frac{\hat{\alpha}_{1,j}}{\|\hat{%
\alpha}_{1,j}\|}= \mathbf{0}_{T\times 1}$, for $j = 1, \ldots, p_T +1$.

For $j = 1, \ldots d$, (i) holds automatically because $\hat{\Upsilon}_0$ is
such that $T^{-1}\Xi_{j}^\top \overline{K} \hat{Z}^o +
p_{\lambda_3}^{^{\prime}} (\| \tilde{B}_{j} \|) \frac{\hat{\alpha}_{0,{j}}}{%
\|\hat{\alpha}_{0,{j}}\|} = \mathbf{0}_{T\times 1}$, and $\hat{\alpha}%
_{0,j_1} = 0$ for $j_1 = d + 1,\ldots, p_T + 1$ and $\hat{\alpha}_{1,j_2}=0$
for $j_2 = d_1 +1 , \ldots, p_T +1$. The same can be said of the
derivative terms of $j = 1, \ldots, d_1$ for (ii). For the other
indices, we have to use the property of the sub-differential to verify the
following conditions: (1) $\max_{d+1\leq j \leq p_T +1} \|T^{-1}\Xi_j^\top \overline{K} \hat{Z}
\| \leq \min_{d+1\leq j \leq p_T +1} p_{\lambda_3}^{^{\prime}} (\| \tilde{B}%
_j \|) $, (2) $\max_{d_1+1\leq j \leq p_T +1} \|T^{-1}\Xi_{j+P_T+1}^{h\top} 
\overline{K} \hat{Z} \| \leq \min_{d_1+\leq j \leq p_T +1}
p_{\lambda_4}^{^{\prime}} (\tilde{D}_j) $.

We focus on (1) since the case for (2) is similar. Note that $\| \tilde{B}_j
\|$ is $O_p(\sqrt{Td}\lambda_1)$ by \thref{fsprop}, since $\hat{\alpha}%
_{0,j} = 0_{T\times 1}$, and Assumption HD.3(ii) implies that $\sqrt{Td}%
\lambda_1 = o(\lambda_3)$. Also, by property of the SCAD penalty, $%
p^{^{\prime}}_{\lambda_3}(\theta) = \lambda_3$ if $0<\theta<\lambda_3$, so
this implies that we have $\min_{d+1\leq j \leq p_T +1}
p_{\lambda_3}^{^{\prime}} (\| \tilde{B}_j \|) = \lambda_3$.

On the other hand, 
\begin{align*}
&\max_{d+1\leq j\leq p_T+1} \|T^{-1}\Xi_j^\top \overline{K} \hat{Z} \| \leq
\max_{d+1\leq j\leq p_T+1} \|T^{-1}\Xi_j^\top \overline{K} \bar{Z} \| +
\max_{d+1\leq j\leq p_T+1} \|T^{-1}\Xi_j^\top \overline{K} (\hat{Z} - \bar{Z}%
) \| \\
&= \max_{d+1\leq j\leq p_T+1} \|T^{-1}\Xi_j^\top \overline{K} \bar{Z} \|  \\ &+
\max_{d+1\leq j\leq p_T+1} \bigg\|T^{-1}\Xi_j^\top \overline{K} \bigg(%
\sum_{l=1}^{p_T+1} \Xi_l (\hat{\alpha}_{0,l}- \alpha_{0,l}) +
\sum_{l=1}^{p_T+1} \Xi_{l+p_T+1}^h h(\hat{\alpha}_{1,l}- \alpha_{1,l}) %
\bigg) \bigg\|,
\end{align*}
where $\overline{Z} = \overline{Y}-\sum_{j=1}^{p_{T}+1}\Xi _{j}\alpha_{0,j}-\sum_{j=1}^{p_{T}+1}\Xi _{j+p_{T}+1}^h h\alpha_{1,j}$. From \eqref{b11}
and \thref{lemmac1}, we know that the first term is $O_p(\sqrt{\log h^{-1}/h}%
)$. For the second term, we use the convergence rates obtained in %
\eqref{convergerate} to get $O_p(\sqrt{d/h})$. Hence, $\max_{d+1\leq j\leq
p_T+1} \|T^{-1}\Xi_j^\top \overline{K} \hat{Z} \| = O_p(h^{-1/2}[(\log
h^{-1})^{1/2} + d^{1/2}]) = o_p(\lambda_3)$ by HD.3(ii).

\end{proof}

\noindent \textbf{Proof of \thref{oracleprop}}

\noindent Let $\hat{\beta}_t^{u}$ be the local linear estimator dervied from solving the oracle problem, which is the estimation problem assuming \textit{a priori} knowledge of the sparsity pattern and excluding penalties. Specifically, let $\hat{\Upsilon}^u$ be the solution to the following:
\begin{equation*}
\underset{\Upsilon \in \mathbb{R}^{T\times d_1}}{\arg
\min }\quad T^{-1} \bigg(\overline{Y}-\sum_{j=1}^{d}\Xi _{j}\alpha
_{0,j}-\sum_{j=1}^{d_1}\Xi_{j+p_T+1}\alpha _{1,j}\bigg)^{\top }%
\overline{K}\bigg(\overline{Y}-\sum_{j=1}^{d}\Xi _{j}\alpha
_{0,j}-\sum_{j=1}^{d_1}\Xi_{j+p_T+1}\alpha _{1,j}\bigg).
\end{equation*}
Then, define $\hat{\beta}_t^u = (\bar{\gamma}_t^{u\top}, \textbf{0}_{1 \times ((p_T+1)-d)})^\top$, where $\bar{\gamma}_t^u=(e_1^\top \otimes I_d) \hat{\gamma}_t^u$, and $\hat{\gamma}_t^u$ is the $t^{th}$ row of $\hat{\Upsilon}^u$. 

Then, we want to show that the two-stage estimate $\hat{\beta}_t^h$ is asymptotically close to the oracle estimator:
\begin{equation}
\label{oracleclose}
\sqrt{Th}\max_{t}\Vert \hat{\beta}_{t}^{h}-\hat{\beta}_{t}^{u}\Vert \rightarrow^{p}
0,
\end{equation}
as $T\rightarrow \infty$.

In light of \thref{ssoracle}, to prove \eqref{oracleclose}, we only need to
show that the biased oracle estimator is asymptotically close to the
unbiased oracle estimator. More formally, let $\hat{\beta}_t^b$ be the $%
t^{th}$ row of $\hat{\Upsilon}_0^o$ be the biased oracle estimator at time $%
t $. We are interested in obtaining an expression for $\max_t \|\hat{\beta}%
_t^b - \hat{\beta}_t^u \|$.

Here, it is much more convenient to use the likelihood-sum
expression of the likelihood similar to \eqref{llscad}. Consider the
optimization for the biased oracle estimator: 
\begin{align*}
&\underset{\Upsilon\in \mathbb{R}^{T \times d_1}}{\arg \min } \ T^{-1}
\sum_{t=1}^T \sum_{s=t-\floor*{Th},s\neq t}^{t+\floor*{Th}} k_{st} \bigg[ %
y_{s+1} - \alpha_{0t}^\top X_s^o - \alpha_{1t}^\top \bigg( \frac{s-t}{T} X_{s,1}^o%
\bigg) \bigg]^2  \notag \\
&+ \sum_{j=1}^{d}p_{\lambda _{3}}^{^{\prime}}\left( \left\Vert \tilde{B}%
_{j}\right\Vert \right) \left\Vert \alpha _{0,j}\right\Vert
+\sum_{j=1}^{d_1}p_{\lambda _{4}}^{^{\prime}}\left( \tilde{D}_{j}\right)
\left\Vert h\alpha _{1,j}\right\Vert, \label{biasedoracle2}
\end{align*}
where $X_s^o =(1, X_{s1}, \ldots, X_{sd})^\top$, and $X_{s,1}^o = (1, X_{s1}, \ldots, X_{sd_1})^\top$.

From this, the $t$-period biased estimate, $\hat{\gamma}_t^b = (\hat{\alpha}%
_{0t}^{b\top}, \hat{\alpha}_{1t}^{b\top})^\top$ can be obtained as 
\begin{equation*}
\hat{\gamma}_t^b = \bigg(T^{-1} \sum_{s=t-\floor*{Th},s\neq t}^{t+\floor*{Th}%
} q^o_{st} q_{st}^{o\top} k_{st} \bigg)^{-1} \bigg[T^{-1} \sum_{s=t-\floor*{%
Th},s\neq t}^{t+\floor*{Th}} q^o_{st} y_{s+1} k_{st} - \check{\lambda}_t %
\bigg]
\end{equation*}
where $q_{st}^{o}=(X_{s1},\ldots ,X_{sd},(\frac{s-t}{T}%
)X_{s1},\ldots ,(\frac{s-t}{T})X_{sd_{1}})^\top$, $\check{\lambda}_t =
[p^{^{\prime}}_{\lambda_3}(\|\tilde{B}_1\|)\hat{\alpha}^b_{0,1,t}/\|\hat{
\alpha}^b_{0,1}\|, \ldots, p_{\lambda _{4}}^{^{\prime}} (\tilde{D}_{d_1}) \\
\hat{\alpha}^b_{1,d_1,t}  /\|\hat{\alpha}^b_{1,d_1}\|]^\top$ is a $(d+d_1)
\times 1$ vector of penalty terms, and $\alpha_{i,j,t}$ refers to the the $%
t^{th}$ element in $\alpha_{i,j}$ for $i=0,1$. As the unbiased oracle estimator does not have the penalty term, we
get $\hat{\gamma}_{t}^b - \hat{\gamma}_{t}^u = \bigg(T^{-1} \sum_{s=t-\floor*{Th}%
,s\neq t}^{t+\floor*{Th}} q^o_{st} q_{st}^{o\top} k_{st}\bigg)^{-1} (-\check{%
\lambda}_t )$.
Then note that 
\begin{equation*}  \label{boracle}
\max_{t} \| \hat{\gamma}_{t}^b - \hat{\gamma}_{t}^u \| \leq \max_t \underbrace{\bigg\|\bigg(T^{-1} \sum_{s=t-\floor*{Th},s\neq
t}^{t+\floor*{Th}} q^o_{st} q_{st}^{o\top} k_{st}\bigg)^{-1}\bigg\|_{2}}_{<
\infty} \|\check{\lambda}_t\|
\end{equation*}
where $\| A \|_2 = \max_{\| x\|=1}\|Ax \|$ for a matrix $A$ and vector $x$,
is the induced matrix norm. Here, we have used the fact that $\|A\|_2$ is
equivalent to its maximum eigenvalue when $A$ is symmetric, and assumption
HD.1(i), which implies that, with probability approaching one, it exists and is finite. Then, notice that 
\begin{equation*}
\|\tilde{B}_j\| \geq \min_{1\leq j \leq (d+d_1)} \| B_j \| - \underbrace{\max_{1
\leq j \leq (d+d_1)} \|\tilde{B}_j - B_j \|}_{= \ O_p(\sqrt{d T} \lambda_1)},
\end{equation*}
where we have used \thref{fsprop} for the second term. Recall that
HD.3(ii) implies $\sqrt{d T} \lambda_1 = o(\lambda_3) = o(\sqrt{T})$.
Furthermore, by HD.3(iii), the first term is larger than $b_\diamond \sqrt{T}
$ and hence it is the leading term. Since $p^{^{\prime}}_{\lambda_3}(\theta)
= \theta$ for $0 < \theta < \lambda_3$ and $\lambda_3 = o(\sqrt{T})$, we
conclude that $\check{\lambda}_t$ vanishes asymptotically at a rate faster
than $\sqrt{T}$. This thus establishes \eqref{oracleclose}.

To establish asymptotic normality, we note that the key difference between the case here and the low-dimensional case of \thref{consis} is the pre-multiplication of $A_T$, particularly to the expression in \eqref{A1}. Hence, we can follow the arguments of \cite{cai2007trending} to obtain asymptotic normality.

\qed  

\end{document}